\pdfoutput=1

\documentclass[11pt]{scrartcl}
\usepackage[a4paper, total={16cm, 24cm}]{geometry}
\usepackage{microtype} 
\usepackage{authblk}
\title{Single-Winner Voting on Matchings}
\date{\vspace{-1.5cm}}

\author[1]{Niclas Boehmer}
\author[1]{Jessica Dierking}
\affil[1]{Hasso Plattner Institute, University of Potsdam, Germany} 
\affil[ ]{\texttt{niclas.boehmer@hpi.de,jessica.dierking@hpi.de}} 

\usepackage[table,xcdraw]{xcolor}
\usepackage{dsfont}
\usepackage{footmisc}
\usepackage{graphicx}
\usepackage{pgfplots}
\pgfplotsset{compat=1.17}
\usepackage{standalone}

\usepackage{tikz}
\usetikzlibrary{calc}
\makeatletter
\tikzset{angle/.code 2 args=\pgfmathsincos{#1}\tikz@lib@place@handle@{#2}{180+#1}{\pgfmathresultx}{\pgfmathresulty}{#1}{1}}
\makeatother

\usepackage{amsmath, amsthm, mathtools, nicefrac, siunitx}
\usepackage[T1]{fontenc}
\usepackage{newtxtext,newtxmath}

\usepackage{thmtools}
\usepackage{thm-restate}
\newtheorem{theorem}{Theorem}[section]

\newtheorem{observation}[theorem]{Observation}
\newtheorem{proposition}[theorem]{Proposition}
\newtheorem{lemma}[theorem]{Lemma}

\theoremstyle{definition}
\newtheorem{definition}[theorem]{Definition}

\usepackage{natbib}
\usepackage{algorithm}
\usepackage[noend]{algorithmic}

\usepackage[inline]{enumitem}
\usepackage{booktabs}
\usepackage{subcaption}
\usepackage{array,hhline,multirow,tabularx,makecell}
\newcolumntype{Y}{>{\centering\arraybackslash}X}

\newcolumntype{K}[1]{>{\centering\arraybackslash\hsize=#1\hsize}X}

\usepackage[disable,textsize=tiny]{todonotes}

\newcommand{\restatehere}[1]{
	\marginline{\vspace{0.6cm}\footnotesize \hyperlink{original#1}{\hypertarget{restated#1}{[Main]}}}
	\csname #1\endcsname*
}
\newenvironment{proofsketch}[1][Proof sketch]{
  \begin{proof}[#1]
  
}{
  \end{proof}
}

\usepackage[pagebackref]{hyperref}
\hypersetup{
	pdfencoding=auto,
	psdextra,
	colorlinks=true,
	citecolor=green!50!black,
	linkcolor=red!60!black,
}

\usepackage[nameinlink]{cleveref}

\crefname{observation}{observation}{observations}
\Crefname{observation}{Observation}{Observations}

\begin{document}

\maketitle

\bigskip
\begin{abstract}
  \begin{center}
    \textbf{\textsf{Abstract}} \smallskip
  \end{center}
  We introduce a single-winner perspective on \emph{voting on matchings}, in which voters have preferences over possible matchings in a graph, and the goal is to select a single collectively desirable matching. Unlike in classical matching problems, voters in our model are not part of the graph; instead, they have preferences over the entire matching.
  In the resulting election, the candidate space consists of all feasible matchings, whose exponential size renders standard algorithms for identifying socially desirable outcomes computationally infeasible.
  We study whether the computational tractability of finding such outcomes can be regained by exploiting the matching structure of the candidate space.
  Specifically, we provide a complete complexity landscape for questions concerning the maximization of social welfare, the construction and verification of Pareto optimal outcomes, and the existence and verification of Condorcet winners under one affine and two approval-based utility models.
  Our results consist of a mix of algorithmic and intractability results, revealing sharp boundaries between tractable and intractable cases, with complexity jumps arising from subtle changes in the utility model or solution concept.
\end{abstract}

\section{Introduction}
Alice is head of the HR department at Condorcet Consulting. Each year, Condorcet Consulting offers several different internship positions and receives a large number of applications. The task of the HR team is to decide which applicant should be assigned to which internship position, with some applicants only being eligible for a subset of the positions. This task becomes all the more difficult because the CEO, CTO, and other members of the management board of the company would like to implement different assignments of the applicants to positions. Therefore, Alice and her team must find a way to aggregate these conflicting opinions into a single assignment.
She quickly realizes that even deciding which principles she should apply when choosing a desirable assignment is not straightforward. What does it mean for one assignment to be better than another? And can one even determine that a proposed assignment has no clearly superior alternative?

The situation reminds her of voting scenarios, in which voters have preferences over candidates and a single compromise candidate must be selected: In her problem, each possible assignment of applicants to positions constitutes a candidate, and the members of the management board act as voters who express preferences over these candidates. Alice describes her problem to her friend Bob, who is an expert in voting. He notices that, while similar in spirit, the problem differs in some ways fundamentally from standard voting scenarios, where the candidate set is typically small and efficiently enumerable. In contrast, the number of possible assignments in Alice's problem grows exponentially with the number of applicants and positions, rendering most classical voting algorithms computationally intractable, as it becomes infeasible to iterate over all candidates. The two are thus left wondering how to aggregate the different preferences and how to do so in a computationally efficient way.

We model Alice's problem as \emph{voting on matchings}: given an underlying graph $G$, select a single matching based on the voters' preferences over possible matchings. The motivating example corresponds to a special case in which $G$ is bipartite, with internship positions on one side and applicants on the other.\footnote{In general, we do not assume the underlying graph to be bipartite, although almost all of our hardness and algorithmic results also hold in this case.} We put forward a new perspective on decision-making in matching under preferences and resource allocation: In the extensive literature on matching under preferences \citep{manlove2013algorithmics,klaus2016matching} and on resource allocation \citep{bouveret2016fair}, agents are typically part of the underlying graph and have preferences over their matched partner or allocated resource. In other words, each agent cares only about the part of the outcome that directly affects them. In contrast, we adopt a centralized perspective where voters represent entities that express preferences over the entire outcome. As a result, our model can capture situations in which multiple central stakeholders---like the management board of a company or teachers in a class---have preferences over the allocation of resources to individuals or the matching of people into~pairs.

In this paper, we focus on the problem of selecting a single winning solution in a voting on matchings instance under three classic notions of social desirability.

\textbf{Social Welfare Optimization:} Selecting candidates that maximize collective utility, evaluated through utilitarian or egalitarian welfare.

\textbf{Pareto Optimality:} Constructing and verifying candidates for which no other candidate is preferred by all voters.

\textbf{Condorcet Criterion:} Proving existence or verifying \emph{Condorcet winners}, i.e., candidates that defeat all others in pairwise majority comparisons.

\subsection{Our Contributions}
We analyze the computational complexity of identifying desirable outcomes when voting on matchings in a given graph. The exponential number of possible matchings renders standard voting algorithms inapplicable. Underlying our study is the central question of whether the combinatorial structure of matchings---the feasibility constraint defining the candidate space---suffices to restore tractability for classical notions of collective optimality.

We provide a complete complexity landscape for three key solution concepts---social welfare, Pareto optimality and Condorcet winners---under three utility models (see \Cref{tab:overview}).\footnote{Our polynomial-time results hold for general graphs. We do not explicitly seek to show NP-hardness for specific graph classes; nonetheless, almost all of our hardness reductions hold on bipartite graphs and simple graph structures like collections of paths and stars.} Each voter specifies a preferred matching, which we lift to utilities under three models, inspired by standard scoring vectors from social choice theory: In the affine utility model, the utility a voter derives from a matching is an affine function of the overlap with the voter's preferred matching, generalizing, for instance, Borda-style scoring. The other two models are approval-based, i.e., a matching provides either zero or one utility: In the one-edge approval model, a voter approves a matching whenever its overlap with the preferred matching is nonempty. In the $\kappa$-missing approval model, parameterized by $\kappa$, a voter approves a matching whenever at most $\kappa$ edges from their preferred matching are missing.

Overall, examining \Cref{tab:overview}, we find that most problems are computationally intractable, with Pareto optimality emerging as the ``most tractable'' of the three solution concepts. Still, we find several cases where we can use the problem's underlying structure to derive polynomial-time results, thereby identifying the boundaries of tractability.
However, the discovered islands of tractability are highly fragile: even minor modifications, such as increasing the parameter $\kappa$ or strengthening the solution concept from weak to strong Condorcet winners, can lead to a sharp jump in complexity.
Returning to our central question, our findings indicate that, in the case of matchings, the combinatorial structure of the candidate space alone does not suffice to restore traceability of classic notions of collective optimality in the face of an exponential candidate space.
We next discuss our results in more detail. Theorems whose proofs are (partly) deferred to the appendix are marked with $(\bigstar)$.

In \Cref{sec:socialwelfare}, we study utilitarian and egalitarian welfare. We show that maximizing utilitarian welfare  is polynomial-time solvable under affine utilities but NP-complete for the approval-based utility models. For egalitarian welfare, we show NP-completeness under affine and one-edge approval utilities. For $\kappa$-missing approval, we show that  maximizing egalitarian welfare is closely related to \textsc{Satisfiability}, allowing us to mirror the complexity jump between \textsc{$2$-Sat} and \textsc{$3$-Sat} to show the same complexity jump between $1$-missing approval, which we show to be polynomial-time solvable, and $2$-missing approval, which we show to be NP-hard.

In \Cref{sec:pareto}, we discuss (weak and strong) Pareto optimality. We show that constructing Pareto optimal candidates is generally easier than verifying if a given candidate is Pareto optimal. In the approval-based settings, we connect the verification of Pareto optimal candidates to the problem of maximizing egalitarian welfare, yielding the same complexity jump as before. In the affine case, we can utilize the feasibility of maximizing utilitarian welfare to construct strongly Pareto optimal candidates, while verifying Pareto optimality turns out to be coNP-complete via a reduction from the corresponding problem in approval-based multiwinner voting.

\Cref{sec:condorcet} deals with Condorcet winners. We show that deciding the existence of Condorcet winners is NP-hard across all of our utility models (except for the trivial existence of weak Condorcet winners under approval-based utilities). The verification of weak and strong Condorcet winners is coNP-complete under all considered utility functions.

\begin{table}[t]
  \scriptsize
  \centering
  \renewcommand{\arraystretch}{1.5}
  \setlength{\tabcolsep}{3pt}
  \setlength{\arrayrulewidth}{0.6pt}
  \setlength{\doublerulesep}{0.5pt}

  \begin{tabularx}{\linewidth}{
    >{\centering\arraybackslash}m{0.22\linewidth}||
    Y|Y|Y|Y
    }
     & \textbf{Affine}                                                                                                                                                                            & \textbf{One-Edge} & \textbf{$\kappa$-missing ($\kappa > 1$)} & \textbf{$\kappa$-missing ($0 \leq \kappa \leq 1$)} \\
    \hhline{=||====}

    \textsc{Utilitarian Welfare}
     & P [\Cref{thm:affine_utilitarian}]
     & \multicolumn{3}{>{\centering\arraybackslash}m{\dimexpr0.15\linewidth*3+6\tabcolsep+2\arrayrulewidth\relax}}{NP-complete [\Cref{thm:approval_utilitarian}]}                                                                                                                                                     \\
    \hhline{-||----}

    \textsc{Egalitarian Welfare}
     & \multicolumn{3}{>{\centering\arraybackslash}m{\dimexpr0.15\linewidth*3+6\tabcolsep+2\arrayrulewidth\relax}|}{NP-complete [\Cref{thm:affine_egalitarian,thm:1_edge_big_kappa_egalitarian}]}
     & P [\Cref{thm:small_kappa_egalitarian}]                                                                                                                                                                                                                                                                         \\
    \hhline{-||----}

    \textsc{wPO-Construction}
     & \multicolumn{4}{>{\centering\arraybackslash}m{\dimexpr0.15\linewidth*4+8\tabcolsep+3\arrayrulewidth\relax}}{P [\Cref{thm:construct_wPO}]}                                                                                                                                                                      \\
    \hhline{-||----}

    \textsc{sPO-Construction}
     & P [\Cref{thm:affine_construct_sPO}]
     & \multicolumn{2}{>{\centering\arraybackslash}m{\dimexpr0.15\linewidth*2+4\tabcolsep+\arrayrulewidth\relax}|}{NP-hard [\Cref{thm:1-edge_big_kappa_construct_sPO}]}
     & P [\Cref{cor:small_kappa_construct_sPO}]                                                                                                                                                                                                                                                                       \\
    \hhline{-||----}

    \textsc{wPO-Verification}
     & \multicolumn{3}{>{\centering\arraybackslash}m{\dimexpr0.15\linewidth*3+6\tabcolsep+2\arrayrulewidth\relax}|}{\multirow{2}{=}{\centering coNP-complete [\Cref{thm:check_wsPO}]}}
     & \multirow{2}{=}{\centering P [\Cref{thm:small_kappa_check_wsPO}]}                                                                                                                                                                                                                                              \\
    \hhline{-||~~~~}
    \textsc{sPO-Verification}
     & \multicolumn{3}{>{\centering\arraybackslash}m{\dimexpr0.15\linewidth*3+6\tabcolsep+2\arrayrulewidth\relax}|}{}
     &                                                                                                                                                                                                                                                                                                                \\
    \hhline{-||----}

    \textsc{wCW-Existence}
     & \multirow{2}{=}{\centering NP-hard [\Cref{thm:affine_exists_wsCW}]}
     & \multicolumn{3}{>{\centering\arraybackslash}m{\dimexpr0.15\linewidth*3+6\tabcolsep+2\arrayrulewidth\relax}}{P [\Cref{thm:approval_exists_wCW}]}                                                                                                                                                                \\
    \hhline{-||~---}
    \textsc{sCW-Existence}
     &
     & \multicolumn{3}{>{\centering\arraybackslash}m{\dimexpr0.15\linewidth*3+6\tabcolsep+2\arrayrulewidth\relax}}{NP-hard [\Cref{thm:approval_exists_sCW}]}                                                                                                                                                          \\
    \hhline{-||----}

    \textsc{wCW-Verification}
     & \multicolumn{4}{>{\centering\arraybackslash}m{\dimexpr0.15\linewidth*4+8\tabcolsep+3\arrayrulewidth\relax}}{\multirow{2}{=}{\centering coNP-complete [\Cref{thm:check_wsCW}]}}                                                                                                                                 \\
    \hhline{-||~~~~}
    \textsc{sCW-Verification}
     & \multicolumn{4}{>{\centering\arraybackslash}m{\dimexpr0.15\linewidth*4+8\tabcolsep+3\arrayrulewidth\relax}}{}                                                                                                                                                                                                  \\
  \end{tabularx}

  \caption{Overview of the results. See \Cref{sec:prelim} for definitions. Results for ``affine'' hold for all affine utility functions and results for ``$\kappa$-missing'' for any $\kappa\in \mathbb{N}_0$.}
  \label{tab:overview}

\end{table}

In \Cref{sec:maximal}, we examine a variant of our problem in which we restrict the candidate space to the set of \emph{maximal matchings}, asking whether this restriction is sufficient to regain tractability. This investigation is motivated by two considerations.
First, implementing a non-maximal matching is arguably unreasonable, as such matchings are clearly inefficient.
Second, natural hardness proofs for many of our problems rely on voters specifying non-maximal matchings---intuitively, because this allows voters to constrain only a small part of the decision. To obtain more robust hardness results and rule out that hardness relies on the presence of such voters, we consider the restriction to maximal matchings.\footnote{In this spirit, most hardness proofs in the appendix directly establish hardness for instances restricted to maximal matchings, thereby implying hardness for the general case. All polynomial-time algorithms for the general case naturally apply to the restricted setting as well.}

Utilizing this restriction, we are able to show that for $\kappa$-missing utilities, the size of the candidate space is effectively rendered polynomial for a fixed $\kappa$ (see \Cref{thm:kappa_maximal_poly}), which yields XP algorithms for all considered computational problems. For the remaining problems, except for the verification of weakly Pareto optimal candidates in the one-edge approval model which becomes trivial, we show that reducing the candidate space to maximal matchings does not regain tractability.

\subsection{Related Work}\label{sub:RW}
The problem of selecting a matching based on agents' preferences has been extensively studied in algorithmic game theory and computational social choice. Foundational work on stable and popular matchings includes \citep{manlove2013algorithmics, gale1962college, cseh2017popular}. More recently, adopting a multiwinner-voting perspective, \citet{boehmer2025proportional} examined the problem of selecting $k$ matchings that proportionally reflect agents’ preferences.
In contrast to most of this literature, in which agents are embedded in the graph and have preferences over the partner(s) to whom they are matched, our approach assumes that voters have preferences over the entire matching.
We model the problem as a single-winner voting scenario with an exponentially large candidate space.

In the following, we discuss two previously studied voting scenarios that deal with exponentially large candidate spaces and have been extensively studied in the literature: multiwinner voting and voting in combinatorial domains.
In \emph{multiwinner voting}, given a set of candidates, the goal is to select a winning committee, i.e., a subset of the candidates, of a given size \citep{lackner2023multi,elkind2017properties}. \emph{Voting in combinatorial domains} considers settings in which possible outcomes are defined as assignments to multiple issues such as choosing a common menu or deciding multiple referenda at the same time \citep{lang2016voting}.

Most work on both models treats outcomes as decomposable collections of individual candidates or decisions, typically assuming additive utilities where the utility of an outcome is the sum of the utilities of its components. Within this framework, various fairness concepts, including different notions of proportionality \citep{masavrik2024generalised,peters2021proportional, aziz2017justified, aziz2018complexity, sanchez2017proportional} have been analyzed. In contrast, our work adopts a \emph{single-winner perspective}, treating each outcome as a structured, indivisible object, and analyzing its properties in comparison with all other feasible outcomes. We assume that each voter casts their top choice and derives utility based on a distance notion to that top choice. \citet{freeman2019truthful} study a similar approach in a participatory budgeting setting.
Our perspective motivates the study of classical notions like Pareto optimality \citep{Luc2008} and the Condorcet criterion \citep{fishburn1977condorcet}, which were originally studied in single-winner contexts. We build on results of \citet{aziz2020computing} and  \citet{darmann2013hard}, who study the complexity of testing these properties for multiwinner voting. Our work extends these ideas by investigating these criteria for more complex feasibility constraints (beyond simple cardinality bounds) and a broader range of utility models (see \Cref{app:RW} for a more detailed discussion).

Closely related to our feasibility constraint, which requires candidates to form matchings in an underlying graph, are recent generalizations of multiwinner voting that impose structural constraints on feasible committees, such as matroid, packing, or matching constraints~\citep{mavrov2023fair, fain2018fair, masavrik2024generalised}.
Our model for voting over matchings can be viewed as a variant of matching-constrained multiwinner voting as studied by \citet{fain2018fair}, where candidates are edges of a graph and feasible committees are required to form a matching. Our work differs from these approaches in the considered preference models and solution concepts. While \citet{fain2018fair} adopt a multiwinner perspective with additive utilities over edges, our model allows for richer utility functions and focuses on classical single-winner criteria such as Pareto optimality, Condorcet winners, and social welfare optimization. In this sense, we extend their framework and related studies primarily centered on proportionality toward a more general analysis of preference aggregation under structural feasibility constraints.

\section{Preliminaries}
\label{sec:prelim}
We will now formally define the voting on matchings model. Let $G = (W,E)$ be a graph. We denote by $\mathcal{C}$ the set of all matchings in $G$, i.e., the subsets of $E$ in which each node is incident to at most one edge. In the following, we refer to matchings as \emph{candidates}. We define a set of \emph{voters} $V$, where each voter $v \in V$ has a preferred matching $M^*(v) \in \mathcal{C}$.\footnote{We assume without loss of generality that each edge in $G$ is in the preferred matching of at least one voter, as otherwise we can simply delete the edge.} We study the problem of selecting a single matching from $\mathcal{C}$ based on the voters' preferences.

\subsection{Utility Functions}
\label{subsec:utility}
The utility that a voter $v \in V$ derives from a candidate $M \in \mathcal{C}$ is defined by a \emph{utility function} $u: V \times \mathcal{C} \rightarrow \mathbb{R}$. We assume that voters prefer candidates from which they derive higher utilities and are indifferent between candidates from which they derive the same utility. We assume throughout that a voter's utility for a matching only depends on the matching's overlap with the voter's preferred matching. Our utility functions are inspired by standard scoring vectors from single-winner voting. We define three different utility models.

First, we define affine utility functions. Herein, a voter gains strictly more utility whenever a candidate has a bigger overlap with their ideal candidate. Two candidates provide the same utility if and only if they have the same overlap size. Formally:
\begin{definition}[Affine Utility]
  Let $M$ be a candidate and $v \in V$ be a voter. Let $\alpha_v \in \mathbb{R}_{> 0}$ and $\beta_v \in \mathbb{R}$.
  An \emph{affine utility function} for $v$ is a utility function of the form
  $u(v, M) = \alpha_v |M^*(v) \cap M| + \beta_v$.
\end{definition}

The coefficients of the utility functions may differ across voters and are treated as part of the voter's representation. That is, a voter $v \in V$ is represented by their preferred matching $M^*(v)$ together with parameters $\alpha_v \in \mathbb{R}_{> 0}$ and $\beta_v \in \mathbb{R}$.

\textit{Remark}: As an extension, one could consider the utility model in which agents assign an individual utility to each edge, and the utility for a matching is the sum of the individual edge utilities. This model is similar to the matching-constrained multiwinner voting studied by \citet{fain2018fair}. In our case, moving to this extension does not change the complexity results, as hardness carries over directly, and our polynomial-time results can be adapted. \smallskip

As a second model, we consider two approval-based settings where voters either approve or disapprove a candidate based on a similarity threshold.  A voter derives utility one if they approve a candidate and utility zero otherwise.

We start by defining the \emph{one-edge approval model}, in which a voter approves a candidate whenever the overlap with their preferred matching is non-empty, i.e., there is at least one common edge between the candidate and the preferred matching.\footnote{Fix some threshold $\kappa\in \mathbb{N}_0$. When considering one-edge approval parameterized by $\kappa$, where voters only approve candidates with an overlap of at least $\kappa$ edges, all our NP-hardness results still hold: There is a simple reduction from one-edge approval, as we can add $\kappa-1$ edges that are approved by all voters. As for most of or considered problems NP-hardness already holds for the one-edge case, we decided to omit the parameterized case.}
\begin{definition}[One-Edge Approval Utility]
  Let $M$ be a candidate and $v \in V$ be a voter. The utility that $v$ derives from $M$ is
  \[u(v, M) :=\begin{cases}
      1 & \text{if }M^*(v) \cap M \neq \emptyset \\
      0 & \text{otherwise.}
    \end{cases}\]
\end{definition}

We next define the \emph{$\kappa$-missing approval model}, parameterized by a threshold $\kappa\in \mathbb{N}_0$, which specifies how much a candidate may deviate from a voter's preferred matching: A voter approves a candidate if at most $\kappa$ edges from their preferred matching are missing.
\begin{definition}[$\kappa$-Missing Approval Utility]
  Let $M$ be a candidate and $v \in V$ be a voter. The utility  that $v$ derives from $M$ is
  \[u(v, M) :=\begin{cases}
      1 & \text{if }|M^*(v) \setminus M| \leq \kappa \\
      0 & \text{otherwise.}
    \end{cases}\]
\end{definition}

\subsection{Solution Concepts}
We are interested in a variety of solution concepts.
\subsubsection*{Social Welfare}
We consider two different classical approaches to social welfare.
The \emph{utilitarian} social welfare defines social welfare as the sum of the individual utilities across all voters\footnote{\label{fn:borda_plurality}When considering affine utilities with $\alpha_v = 1$ and $\beta_v = 0$ maximizing the utilitarian welfare directly corresponds to computing a Borda winner, assuming that a voter gives candidates in the same equivalence class the same score. Maximizing utilitarian welfare in the $0$-missing approval setting corresponds to finding a Plurality winner.}, giving rise to the following computational problem:
\begin{definition}[\textsc{Utilitarian Welfare}]
  \label{def:decision_social_score}
  Given candidates $\mathcal{C}$, voters $V$, and a number $k \in \mathbb{R}$, we ask whether there exists a candidate $c \in \mathcal{C}$ whose utilitarian welfare is at least $k$, i.e., does $\max_{c \in \mathcal{C}}\sum_{v \in V} u(v, c) \geq k$ hold?
\end{definition}

In contrast, the \emph{egalitarian} social welfare defines the social welfare as the minimum utility across all voters.
\begin{definition}[\textsc{Egalitarian Welfare}]
  \label{def:decision_egalitarian_welfare}
  Given candidates $\mathcal{C}$, voters $V$, and a number $k \in \mathbb{R}$, we ask whether there exists a candidate $c \in \mathcal{C}$ with egalitarian welfare of at least $k$, i.e., does $\max_{c \in \mathcal{C}}\min_{v \in V} u(v, c) \geq k$ hold?
\end{definition}

\subsubsection*{Pareto Optimality}
A Pareto improvement for a candidate $c$ formalizes the idea that there exists a candidate $c'$ such that no voter would object to switching from $c$ to $c'$. A candidate is called Pareto efficient or Pareto optimal if no such improvement exists \citep{Luc2008}.
We distinguish two types of Pareto improvements, leading to two notions of Pareto optimality.

A candidate $c'\in \mathcal{C}$ is a \textit{strong Pareto improvement} for a candidate $c\in \mathcal{C}$ if all voters derive a strictly higher utility from $c'$ than from $c$. A candidate for which no such improvement exists is called \textit{weakly Pareto optimal}.
A candidate $c'\in \mathcal{C}$ is a \textit{weak Pareto improvement} for a candidate $c\in \mathcal{C}$ if at least one voter derives strictly higher utility from $c'$ than from $c$, while no voter derives a strictly smaller utility from $c'$ than from $c$. A candidate for which no such improvement exists is called \textit{strongly Pareto optimal}.
Note that every strong improvement is also a weak improvement and thus strong Pareto optimality implies weak Pareto optimality. By definition, weak and strong Pareto optimal candidates always exist, but are not necessarily unique. We define the problems of finding and verifying Pareto optimal candidates:

\begin{definition}[\textsc{w(s)PO-Construction}]
  \label{def:find_wsPO}
  Given candidates $\mathcal{C}$ and voters $V$, find a candidate $c \in \mathcal{C}$ that is weakly (strongly) Pareto optimal.
\end{definition}

\begin{definition}[\textsc{w(s)PO-Verification}]
  \label{def:check_wsPO}
  Given candidates $\mathcal{C}$, voters $V$, and a candidate $c \in \mathcal{C}$, we ask whether $c$ is weakly (strongly) Pareto optimal.
\end{definition}

\subsubsection*{Condorcet Winners}
A \emph{Condorcet winner} is a candidate that defeats every other candidate in a pairwise majority comparison. In such a comparison, we count the number of voters who strictly prefer each candidate, ignoring those who are indifferent: Formally, a candidate $c\in \mathcal{C}$ wins a pairwise comparison against another candidate $c'\in \mathcal{C}$ if more voters prefer $c$ to $c'$ than the other way around, and ties if the numbers are equal. We distinguish between \emph{weak Condorcet winners}, who win or tie the pairwise comparison against every other candidate, and \emph{strong Condorcet winners}, who strictly win all comparisons \citep{fishburn1977condorcet}.

Unlike Pareto optimal candidates, Condorcet winners may not exist. We define the problems of deciding if a Condorcet winner exists and verifying if a given candidate is a Condorcet winner.

\begin{definition}[\textsc{w(s)CW-Existence}]
  \label{def:exists_sCW}
  Given candidates $\mathcal{C}$, voters $V$, we ask whether a weak (strong) Condorcet winner exists.
\end{definition}

\begin{definition}[\textsc{w(s)CW-Verification}]
  \label{def:check_sCW}
  Given candidates $\mathcal{C}$, voters $V$, and a candidate $c \in \mathcal{C}$, we ask whether $c$ is a weak (strong) Condorcet winner.
\end{definition}

\section{Social Welfare}
\label{sec:socialwelfare}
In this section, we study the problems \textsc{Utilitarian Welfare} and \textsc{Egalitarian Welfare}. We demonstrate that their computational complexity decisively depends on the underlying utility model. Specifically, we provide polynomial-time algorithms for \textsc{Utilitarian Welfare} under affine utilities and for \textsc{Egalitarian Welfare} under $0$-missing and $1$-missing approval. For all remaining cases, we establish NP-completeness.

\subsection{Affine Utility Functions}
In the affine setting, the utility that a voter $v$ gains from a single edge is either $\alpha_v$ if the edge is part of the voter's preferred matching or zero otherwise. Thus, the contribution of an edge to a candidate's utilitarian welfare depends solely on the voters' preferred matchings and their $\alpha$ values, but not on the remaining selected edges. This property allows us to maximize utilitarian welfare in polynomial time by finding a maximum-weight matching.
\begin{theorem}
  \label{thm:affine_utilitarian}
  \textsc{Utilitarian Welfare} under affine utility functions is solvable in polynomial time.
\end{theorem}
\begin{proof}
  We test whether a matching with utilitarian welfare of at least $k$ exists by finding a matching with maximum utilitarian welfare:
  Let $G = (W, E)$ be a graph, and let $V$ be a set of voters. For an edge $e \in E$ and voter $v \in V$, let the indicator $\mathds{I}(e, v)$ be one if and only if $e \in M^*(v)$. We introduce a weight function on the edges of~$G$:
  $w: E \rightarrow \mathbb{R} \text{ with } w(e) = \sum_{v \in V} \mathds{I}(e, v) \alpha_v.$
  The utilitarian welfare of a matching $M$ is
  \[
    \sum_{v \in V} u(v, M) = \sum_{v \in V} \alpha_v |M^*(v) \cap M| + \sum_{v \in V} \beta_v.
  \]
  Note that the second sum is independent of the matching $M$. Thus, to find a matching with maximum utilitarian welfare, we only need to maximize the first sum. Using the definition of the weight function, we obtain:
  \[
    \sum_{v \in V} \alpha_v |M^*(v) \cap M| = \sum_{v \in V}\sum_{e \in M} \alpha_v \mathds{I}(e,v) = \sum_{e \in M} w(e).
  \]
  Therefore, a matching that maximizes utilitarian welfare corresponds to a maximum-weight matching under $w$. Such a matching can be found using, for example, Edmonds' algorithm~\citep{edmonds1965maximum}, which runs in $\mathcal{O}(|W|^2|E|)$.
\end{proof}

In contrast, the utility an edge contributes to a matching's egalitarian welfare depends on the other selected edges, as the other edges influence whether a voter's utility for a matching exceeds $k$. It turns out that this lack of ``decomposability'' renders the problem NP-hard:
\begin{restatable}[$\bigstar$]{theorem}{affineegalitarian}
  \label{thm:affine_egalitarian}
  \textsc{Egalitarian Welfare} under affine utility functions is NP-complete.
\end{restatable}
\begin{proofsketch}
  To establish this result, we reduce from \textsc{$3$-Sat} by constructing a graph composed of gadgets encoding the variables of the \textsc{$3$-Sat} instance, with matchings corresponding to truth assignments. The voters represent the clauses of the instance, where each voter derives positive utility from a matching if and only if the corresponding assignment satisfies the clause. Hence, a candidate with positive egalitarian welfare corresponds directly to a fulfilling assignment of the \textsc{$3$-Sat} instance.
\end{proofsketch}

\subsection{Approval-based Utility Functions}
In the approval-based setting, the interpretation of utilitarian and egalitarian welfare changes slightly, since the only possible utilities a matching can provide for a voter are zero and one. Thus, a candidate's utilitarian welfare is simply the number of voters that approve the candidate, while the egalitarian welfare is one if and only if \emph{all} voters approve the candidate. We show that \textsc{Egalitarian Welfare} can be solved in polynomial time under $0$-missing approval and $1$-missing approval, but remains NP-complete for all other utility models. The hardness of \textsc{Utilitarian Welfare} follows directly when \textsc{Egalitarian Welfare} is hard. We show that \textsc{Utilitarian Welfare} is NP-complete in the other cases as well.

Finding a candidate with positive egalitarian welfare corresponds to finding a candidate approved by all voters. Using a similar idea as in \Cref{thm:affine_egalitarian}, we connect \textsc{Egalitarian welfare} under $\kappa$-missing approval to \textsc{$\kappa+1$-Sat}. Mirroring the complexity jump between \textsc{$2$-Sat} and \textsc{$3$-Sat}, we show that \textsc{Egalitarian welfare} is solvable in polynomial time for $\kappa \in \{0,1\}$, but becomes NP-complete for $\kappa > 1$. We start by establishing the positive result:
\begin{restatable}[$\bigstar$]{theorem}{smallegalitarian}
  \label{thm:small_kappa_egalitarian}
  \textsc{Egalitarian Welfare} under $\kappa$-missing approval for $\kappa \leq 1$ is solvable in polynomial time.
\end{restatable}
\begin{proof}[Proof ($\kappa = 1$)]
  Let $G = (W,E)$ be the input graph and $V$ the set of voters. We construct a corresponding \textsc{$2$-Sat} instance with $|E|$ variables and $\mathcal{O}(|V| \cdot |E|^2)$ clauses.
  For each edge $e \in E$, we introduce a variable $x_e$. We add the following types of clauses:
  \begin{itemize}
    \item For every edge $e \in E$ and each edge $\hat e$ adjacent to $e$, we add the clause $(\overline{x_e} \vee \overline{x_{\hat e}})$, ensuring the matching constraint. Formally, let $C_M := \{(\overline{x_e} \vee \overline{x_{\hat e}}) \mid e \in E, \hat e \in N(e)\}$, where $N(e)$ are the edges that share an endpoint with $e$.
    \item For each voter $v \in V$, we add a clause for every pair of distinct edges in their preferred matching. That is, for $e \neq\hat e \in M^*(v)$, we add clause $(x_e \vee x_{\hat e})$. Formally, let $C_v := \{(x_e \vee x_{\hat e}) \mid e\neq\hat e \in M^*(v)\}$ and $C_{V} := \bigcup_{v \in V} C_v$.
  \end{itemize}
  The complete set of clauses for the \textsc{$2$-Sat} formula is $C := C_M \cup C_{V}$.
  We solve this \textsc{$2$-Sat} instance using a linear-time algorithm, such as the one by \citet{linear-time-algorithm-boolean-aspvall}.
  If the \textsc{$2$-Sat} instance is satisfiable, we return Yes; otherwise, No.

  We now show that a candidate with egalitarian welfare one corresponds to a satisfying assignment, and vice versa.

  $[\Rightarrow]$
  Assume there exists a candidate $M$ with an egalitarian welfare of one. That is, $M$ satisfies the matching constraint, and for every voter $v \in V$, at most one edge of $M^*(v)$ is missing from $M$. Let $x_e = \textit{true}$ if and only if $e \in M$. Then:
  \begin{itemize}
    \item Each clause $(\overline{x_e} \vee \overline{x_{\hat e}}) \in C_M$ is satisfied because $M$ contains at most one of any pair of adjacent edges.
    \item For each voter $v$, every clause $(x_e \vee x_{\hat e}) \in C_v$ is satisfied because at most one edge from $M^*(v)$ is excluded from $M$.
  \end{itemize}

  $[\Leftarrow]$
  Assume the \textsc{$2$-Sat} instance is satisfiable. Let $M := \{e \mid x_e \text{ is assigned \textit{true}}\}$. Then:
  \begin{itemize}
    \item $M$ satisfies the matching constraint, since every pair of adjacent edges $e$ and $\hat e$ is not simultaneously included, as otherwise $(\overline{x_e} \vee \overline{x_{\hat e}})$ would not be satisfied.
    \item For each voter $v \in V$, the clauses $C_v$ ensure that at most one edge from $M^*(v)$ is excluded from $M$: If there were two edges $e\neq\hat{e}\in M^*(v)\setminus M$, the clause $(x_e\vee x_{\hat e})$ would not be satisfied. Therefore, $v$ approves $M$, and the candidate has egalitarian welfare of one.
  \end{itemize} \vspace{-0.5cm}
\end{proof}

In contrast, for $\kappa$-missing approval with $\kappa > 1$ and for one-edge approval, we can establish hardness using a similar idea as for \textsc{Egalitarian Welfare} under affine utilities.
\begin{restatable}[$\bigstar$]{theorem}{oneedgebigegalitarian}
  \label{thm:1_edge_big_kappa_egalitarian}
  \textsc{Egalitarian Welfare} under one-edge approval and under $\kappa$-missing approval with $\kappa > 1$ is NP-complete.
\end{restatable}
\begin{proofsketch}
  The reduction idea is similar to \Cref{thm:affine_egalitarian}. For $\kappa$-missing approval, we reduce from \textsc{$\kappa+1$-Sat}, and for one-edge approval, from \textsc{$3$-Sat}. We construct gadgets corresponding to the variables of the \textsc{Sat} instance and voters corresponding to the clauses. A voter approves a candidate if at least one of the edges of the respective clause is selected. Therefore, a matching with egalitarian welfare of one directly translates to a fulfilling assignment of the \textsc{Sat} instance.
\end{proofsketch}

For \textsc{Utilitarian Welfare}, NP-hardness for one-edge and $\kappa$-missing approval with $\kappa > 1$ follows directly from the NP-hardness of \textsc{Egalitarian Welfare} established in \Cref{thm:1_edge_big_kappa_egalitarian}: In the approval-based setting, a candidate has an egalitarian welfare of one if and only if it has a utilitarian welfare of $|V|$. In contrast, for $0$-missing approval and $1$-missing approval, we cannot extend the polynomial-time solvability of \textsc{Egalitarian Welfare} from \Cref{thm:small_kappa_egalitarian} to \textsc{Utilitarian Welfare}.\footnote{The \textsc{$2$-Sat} algorithm from before would become a variant of \textsc{MAX $2$-Sat} which is NP-complete \citep{garey1974some}.} We show NP-hardness via a reduction from \textsc{Independent Set}.
\begin{restatable}[$\bigstar$]{theorem}{approvalutilitarian}
  \label{thm:approval_utilitarian}
  For any $\kappa \in \mathbb{N}_0$, \textsc{Utilitarian Welfare} is NP-complete under $\kappa$-missing approval and under one-edge approval.
\end{restatable}
\begin{proofsketch}
  For $0$-missing approval, we reduce from \textsc{Independent Set}. We construct an instance such that a candidate approved by $k$ voters corresponds to an independent set of size $k$. Specifically, we add a voter for each node in the input graph $\tilde G$. For each edge in $\tilde{G}$, we add a gadget consisting of two connected edges, each corresponding to one endpoint of the original edge in $\tilde{G}$. The gadgets are not connected to each other and pairwise disjoint. A voter corresponding to node $i$ approves all edges that correspond to node $i$. Under $0$-missing approval, a voter approves a matching if and only if all of their corresponding edges are included. For the correctness, observe that two voters corresponding to adjacent nodes in $\tilde{G}$ cannot both approve the same matching, since two edges from their preferred matchings share an endpoint.

  For $1$-missing approval, we add a separate control gadget and additional voters such that any candidate including a specific ``blocking'' edge cannot reach the required utility threshold. This enforces behavior equivalent to the $\kappa = 0$ case.
\end{proofsketch}

\section{Pareto Optimality}
\label{sec:pareto}
We now turn to  Pareto optimality. Here, only voters' ordinal preferences over candidates matter. We show that for affine utility functions, the ordinal preference relation is independent of the concrete coefficients of the utility function.
\begin{restatable}[$\bigstar$]{lemma}{lemmascoreequivalence}
  \label{lem:affine_score_equivalence}
  All affine utility functions agree on whether a voter $v \in V$ prefers candidate $M$ over $M'$.
\end{restatable}
This property holds as voters prefer matchings with larger overlap with their preferred matching under all affine utility functions, and are indifferent between two matchings with the same overlap size.
\Cref{lem:affine_score_equivalence} allows us to disregard the specific coefficients of the utility function in the affine setting. Throughout this section, we assume for each voter $v \in V$ that $\alpha_v = 1$ and $\beta_v = 0$.

\subsection{Verification}
Recall that our goal is to determine whether a given candidate is Pareto optimal. It turns out that verifying weak and strong Pareto optimality is solvable in polynomial time for $0$-missing  and $1$-missing approval, but coNP-complete for all other utility models.

To show hardness, for affine utilities, we can reduce from verifying the Pareto optimality of a given committee in an approval-based multiwinner election, which was shown to be NP-hard by  \citet{aziz2020computing}. In the approval-based settings, the question of whether a candidate is Pareto optimal is closely related to the question of whether there is a matching that is approved by every voter, i.e., that has positive egalitarian welfare. Utilizing this connection, we show coNP-completeness of \textsc{w/sPO-Verification} for one-edge approval and $\kappa$-missing approval with $\kappa > 1$. Due to technical reasons, we reduce from \textsc{$\kappa+1$-Sat} for \textsc{sPO-Verification} under $\kappa$-missing approval and from the corresponding \textsc{Egalitarian Welfare} problem in the other cases:
\begin{restatable}[$\bigstar$]{theorem}{verificationhardness}
  \label{thm:check_wsPO}
  \textsc{wPO-Verification} and \textsc{sPO-Verification} under affine utility functions, one-edge approval and $\kappa$-missing approval with $\kappa > 1$ are coNP-complete.
\end{restatable}

For $\kappa$-missing approval with $\kappa \leq 1$, we use the polynomial-time algorithms from \Cref{thm:small_kappa_egalitarian} for \textsc{Egalitarian Welfare} to search for Pareto improvements by checking if there are candidates that are approved by all currently approving and at least one additional voter. This tractability stands in contrast with the NP-hardness under the other utility models, further highlighting the jump in complexity between $1$-missing approval and $2$-missing~approval.
\begin{restatable}{theorem}{smallverification}
  \label{thm:small_kappa_check_wsPO}
  \textsc{wPO-Verification} and \textsc{sPO-Verification} under $\kappa$-missing approval with $\kappa \leq 1$ are solvable in polynomial time.
\end{restatable}
\begin{proof}
  For \textsc{sPO-Verification}, given a graph $G$, a set of voters $V$, and a candidate $M$, let $V_M$ be the set of voters that approve $M$. For each voter $v \in V \setminus V_M$ we can test if there is a matching in graph $G$ that is approved by all voters in $V_M \cup \{v\}$ by checking if there is a matching with positive egalitarian welfare for graph $G$ and voters $V_M \cup \{v\}$ using the polynomial-time algorithms from  \Cref{thm:small_kappa_egalitarian}.
  If such a matching exists for some voter $v \in V \setminus V_M$, we return No; otherwise, we return Yes.

  For \textsc{wPO-Verification}, let $G$ be a graph, $V$ a set of voters, and $M$ a candidate.
  If $M$ is approved by at least one voter, then $M$ is trivially weakly Pareto optimal, as that voter cannot improve. Otherwise, finding an improving matching is equivalent to finding a matching that is approved by all voters, i.e., a matching with an egalitarian welfare of one, for which we can again use the algorithm described in \Cref{thm:small_kappa_egalitarian} to test if such a matching exists.
\end{proof}

\subsection{Construction}
We now turn to the construction of Pareto optimal candidates.
We show that finding Pareto optimal candidates is often easier than verifying whether a given candidate is Pareto optimal. This phenomenon is not uncommon, as Pareto optimal candidates are generally not unique, and in some settings, they can be identified efficiently by exploiting simple structural properties.
It is trivial to find a weakly Pareto optimal candidate: We can simply select a candidate that provides maximum utility to some voter. Since this voter cannot strictly prefer any other candidate, the selected candidate is weakly Pareto optimal.
\begin{observation}
  \label{thm:construct_wPO}
  \textsc{wPO-Construction} is solvable in linear time for all considered utility functions.
\end{observation}

Finding strongly Pareto optimal candidates is less straightforward, as a dominating candidate only requires one voter to strictly improve. Depending on the utility model, the problem remains solvable in polynomial time or becomes NP-hard. We first discuss the polynomial cases: affine utility functions and $\kappa$-missing approval with $\kappa \leq 1$.

\begin{proposition}
  \label{thm:affine_construct_sPO}
  \textsc{sPO-Construction} under affine utility functions is solvable in polynomial time.
\end{proposition}
\begin{proof}
  Under affine utility functions, any candidate with maximum utilitarian welfare is strongly Pareto optimal, since a dominating candidate would have strictly higher utilitarian welfare.
  Therefore, a candidate with maximum utilitarian welfare, which can be found in polynomial time (see \Cref{thm:affine_utilitarian}), is strongly Pareto optimal.
\end{proof}

In the approval-based setting, a Pareto improvement requires at least one additional voter to approve the candidate, while preserving all existing approvals.
The construction of strongly Pareto optimal candidates for $0$-missing and $1$-missing approval utilizes the polynomial-time solvability of \textsc{sPO-Verification} (\Cref{thm:small_kappa_check_wsPO}) by iteratively finding Pareto improvements until the current candidate is Pareto optimal.
\begin{proposition}
  \label{cor:small_kappa_construct_sPO}
  \textsc{sPO-Construction} under $\kappa$-missing approval with $\kappa \leq 1$ is solvable in polynomial time.
\end{proposition}
\begin{proof}
  Given a graph $G$ and a set of voters $V$, we describe an iterative algorithm for finding a strongly Pareto optimal candidate. We start with the empty matching. Using \Cref{thm:small_kappa_check_wsPO}, we can check in polynomial time whether a given matching is strongly Pareto optimal. Note that for a matching $M$ that is not strongly Pareto optimal, the algorithm described in \Cref{thm:small_kappa_check_wsPO} returns a matching $M'$ that is a weak Pareto improvement over $M$. We iteratively repeat this process, using $M'$ as the new candidate, until we reach a candidate that is strongly Pareto optimal. Since the number of voters approving the candidate increases by at least one in each step, this process requires at most $|V|$ steps.
\end{proof}

Similar to the hardness of \textsc{sPO-Verification} for one-edge approval and $\kappa$-missing approval with $\kappa > 1$, we use the corresponding hardness results for \textsc{Egalitarian Welfare} (\Cref{thm:1_edge_big_kappa_egalitarian}) to show the hardness of \textsc{sPO-Construction}.
\begin{restatable}[$\bigstar$]{proposition}{approvalconstructspo}
  \label{thm:1-edge_big_kappa_construct_sPO}
  \textsc{sPO-Construction} under one-edge approval and under $\kappa$-missing approval with $\kappa > 1$ is NP-hard.
\end{restatable}
\begin{proofsketch}
  The idea is to use \textsc{sPO-Construction} as a subroutine to solve \textsc{Egalitarian Welfare}. In our reduction, if the candidate returned by \textsc{sPO-Construction} is approved by all voters, it has an egalitarian welfare of one. Otherwise, no candidate with positive egalitarian welfare exists, as such a candidate would Pareto dominate the returned candidate.
\end{proofsketch}

\section{Condorcet Winners}
\label{sec:condorcet}
We now investigate the problem of identifying Condorcet winners, i.e., candidates that win (or tie) the pairwise comparison against every other candidate. Similar to Pareto optimality, only the ordinal preferences of voters over candidates matter when determining Condorcet winners. Thus, by \Cref{lem:affine_score_equivalence}, for affine utilities we can disregard the specific coefficients of the utility function and again assume $\alpha_v = 1$ and $\beta_v = 0$ for each voter $v \in V$.

In the approval-based setting, we observe the following useful relationship between Condorcet winners and candidates with maximum utilitarian welfare:
\begin{lemma}[folklore]
  \label{lem:condorcet_social_score}
  In the approval-based setting, a candidate $M$ is a weak Condorcet winner if and only if it has maximum utilitarian welfare. If no other candidate achieves the same utilitarian welfare, then $M$ is also a strong Condorcet winner.
\end{lemma}
\begin{proof}
  In the approval-based voting setting, each voter assigns either utility zero or one to a candidate. Hence, the utilitarian welfare of a matching equals the number of voters who approve it.

  Assume candidate $M$ has maximum utilitarian welfare. Let $M'$ be any other candidate. If the utilitarian welfare of $M'$ is lower than that of $M$, then fewer voters approve $M'$ than $M$, implying that $M'$ loses the pairwise comparison against $M$. If the welfare values are equal, the candidates tie.

  Conversely, suppose $M$ is a weak Condorcet winner, and let $M'$ be any other candidate. Since $M$ wins or ties against $M'$ in a pairwise comparison, the number of voters who prefer $M'$ over $M$ is less than or equal to the number who prefer $M$ over $M'$. Therefore, the number of voters who approve $M$ is greater than or equal to those who approve $M'$, implying that $M$ has maximum utilitarian welfare. In the strong case, the inequality is strict for all $M' \neq M$ and $M$ is the unique candidate with maximum utilitarian welfare.
\end{proof}

\subsection{Verification}
Verifying whether a given candidate is a Condorcet winner is coNP-complete under all considered utility models. The proofs use reductions from \textsc{$3$-Sat} and the respective \textsc{Utilitarian Welfare} problems in the approval-based settings, and from \textsc{wPO-Verification} in the affine setting.
\begin{restatable}[$\bigstar$]{theorem}{cwverification}
  \label{thm:check_wsCW}
  \textsc{wCW-Verification} and \textsc{sCW-Verification} are coNP-complete under affine utilities, one-edge approval and $\kappa$-missing approval for any $\kappa\in \mathbb{N}_0$.
\end{restatable}
\begin{proofsketch}
  For the affine setting, we reduce from \textsc{wPO-Verification}. The idea is to add extra voters to the given instance who strictly prefer the given candidate $M$ over any other candidate. These voters receive strictly lower utility from any other candidate, ensuring that another candidate $M'$ can only win or tie the pairwise comparison against $M$ if all original voters strictly prefer $M'$ over $M$, i.e., if $M'$ constitutes a strong Pareto improvement for $M$.

  In the approval-based setting, we leverage \Cref{lem:condorcet_social_score} to construct reductions from \textsc{Utilitarian Welfare} and \textsc{$3$-Sat}. For $\kappa$-missing approval we reduce from \textsc{Utilitarian Welfare}. We modify the instance by inserting a candidate $M$ approved by additional voters who approve no other candidate and ask whether $M$ is a Condorcet winner.
  The instance is constructed such that if $M$ is not a Condorcet winner, then there exists a candidate with utilitarian welfare of $k$ in the original instance. In the one-edge approval setting we reduce from \textsc{$3$-Sat}, where we use a similar idea as before: We construct gadgets for the variables and voters for the clauses such that a voter approves a candidate if and only if the corresponding clause is satisfied. Using additional control voters we construct a candidate that is a Condorcet winner unless there exists a fulfilling assignment for the \textsc{$3$-Sat} instance.
\end{proofsketch}

\subsection{Existence}
In contrast to Pareto optimal candidates, Condorcet winners are not guaranteed to exist. We therefore study the computational complexity of deciding whether a Condorcet winner exists in a given instance. We first show that \textsc{wCW-Existence} and \textsc{sCW-Existence} are coNP-hard in the affine setting. Then, we turn to the approval-based setting, where only \textsc{sCW-Existence} remains coNP-hard, while \textsc{wCW-Existence} becomes trivial.

To show that \textsc{wCW-Existence} and \textsc{sCW-Existence} under affine utilities are coNP-hard, we can use an idea very similar to the reduction in \Cref{thm:check_wsCW}. We reduce from \textsc{wPO-Verification} and transform the instance such that the input candidate $M$ is the Condorcet winner or no Condorcet winner exists.
\begin{restatable}[$\bigstar$]{theorem}{affinecondorcetexist}
  \label{thm:affine_exists_wsCW}
  \textsc{wCW-Existence} and \textsc{sCW-Existence} under affine utility functions are coNP-hard.
\end{restatable}

In the approval-based setting, a weak Condorcet winner is guaranteed to exist: As shown in \Cref{lem:condorcet_social_score}, each candidate with maximum utilitarian welfare is a weak Condorcet winner, and such candidates are naturally guaranteed to exist.
\begin{observation}
  \label{thm:approval_exists_wCW}
  \textsc{wCW-Existence} is solvable in constant time for both one-edge approval and $\kappa$-missing approval for any $\kappa\in \mathbb{N}_0$.
\end{observation}

In contrast, the existence of strong Condorcet winners remains coNP-hard; as established in \Cref{lem:condorcet_social_score}, they require the candidate with maximum utilitarian welfare to be unique.
To show hardness, we can use a construction very similar to the reductions in \Cref{thm:check_wsCW}. We construct a graph in which either the candidate with maximum utilitarian welfare or the solution of the \textsc{Sat}-instance is the Condorcet winner or no Condorcet winner exists.
\begin{restatable}[$\bigstar$]{theorem}{approvalsCWexist}
  \label{thm:approval_exists_sCW}
  \textsc{sCW-Existence} under one-edge approval and $\kappa$-missing approval for any $\kappa\in \mathbb{N}_0$ is coNP-hard.
\end{restatable}

\section{Maximal Matchings}
\label{sec:maximal}
In this section, we analyze a variant of voting on matchings in which we restrict the candidate space to maximal matchings. A matching is \emph{maximal} if no additional edge can be added without violating the matching property. We ask whether this restriction suffices to regain tractability for some of our solution concepts. Our investigation is motivated by the fact that this restriction prunes the candidate space by eliminating ``inefficient'', and thus suboptimal, candidates.
While we believe the maximal matching assumption is generally plausible, we also note that restricting to maximal matchings is not always natural or intuitive. For instance, assume we are in a resource allocation setting with three agents $\{A, B, C\}$ and three resources $\{1,2,3\}$, and a voter $v$ wants agent $A$ to get resource $1$ and agent $B$ to get resource $2$. If the voter does not care at all about agent $C$, forcing them to include edge $(C,3)$ in their preferred matching feels counterintuitive, as our model would treat the satisfaction $v$ derives from $(C,3)$ the same as for $(A,1)$ or $(B,2)$. Especially in the approval-based setting, where the inclusion of a single additional edge can decide between approval and disapproval, this restriction might be too strong.

As this setting is a special case of our general setting, polynomial-time algorithms apply directly. For each question, we either show that it becomes feasible, which we will detail in the following, or our NP-hardness reductions in the appendix hold in the case of maximal matchings. An overview of results for the restricted case can be found in \Cref{tab:maximal_overview}.

We present the two cases in which tractability is regained. We first show that under $\kappa$-missing approval, restricting to maximal matchings makes the number of \emph{relevant} candidates, i.e., candidates that are approved by at least one voter, polynomial for fixed $\kappa$.\footnote{In the general setting, no such bound is possible, since a voter's preferred matching may contain few edges; in that case there can even be exponentially many matchings that contain a preferred matching.} Consequently, standard algorithms that iterate over all relevant candidates yield XP algorithms parameterized by $\kappa$ for all studied~problems.
\begin{theorem}
  \label{thm:kappa_maximal_poly}
  For any $\kappa \in \mathbb{N}_0$, under $\kappa$-missing approval on maximal matchings,  the number of maximal matchings that are approved by at least one voter is in $\mathcal{O}\left(|V| |E|^{3\kappa}\right)$.
\end{theorem}
\begin{proof}
  Consider a graph $G$, a set of voters $V$, and a fixed $\kappa \in \mathbb{N}_0$. For an arbitrary voter $v \in V$, we want to bound the number of matchings that are approved by $v$.
  A matching $M$ is approved by $v$ if at most $\kappa$ edges from $M^*(v)$ are not in $M$.
  Let $k \leq \kappa$ and let $M$ be a maximal matching with $|M^*(v) \setminus M| = k$. We aim to bound $|M \setminus M^*(v)|$. Since $M^*(v)$ is maximal, any edge in $M \setminus M^*(v)$ must be adjacent to one of the $k$ edges from $M^*(v) \setminus M$. Each such edge has two endpoints, so $|M \setminus M^*(v)|\leq 2k$.

  There are $\binom{|M^*(v)|}{k}$ ways to choose $M^*(v) \setminus M$. For each of the $2k$ endpoints of the edges in $M^*(v) \setminus M$, we can choose from at most $|E|$ edges, yielding an upper bound of $|E|^{2k}$ possibilities for the edges in $M \setminus M^*(v)$. Thus, the total number of matchings that are approved by $v$ can be upper bounded by $\sum_{i=0}^\kappa \binom{|M^*(v)|}{i} |E|^{2i}$,
  which is in $\mathcal{O}(|E|^{3\kappa})$. Since there are $|V|$ voters, the total number of matchings approved by any voter is in $\mathcal{O}(|V| |E|^{3\kappa})$.
\end{proof}

Further, in contrast to the general case, \textsc{wPO-Verification} under one-edge approval utilities becomes trivial when restricting the candidate space to maximal matchings.\footnote{This result only holds since we assume that each edge in $G$ is approved by at least one voter. Otherwise, the problem would remain coNP-hard as we could reduce from \textsc{Egalitarian Welfare} similar to the unrestricted case.} As (apart from the trivial case where $G$ is empty) every maximal matching contains at least one edge, it is approved by at least one voter and therefore is weakly Pareto optimal.
\begin{observation}
  \label{thm:1-edge_check_wPO_maximal}
  \textsc{wPO-Verification} on maximal matchings under one-edge approval is solvable in constant time.
\end{observation}

\begin{table}[t]
  \scriptsize
  \centering
  \renewcommand{\arraystretch}{1.5}
  \setlength{\tabcolsep}{3pt}
  \setlength{\arrayrulewidth}{0.6pt}
  \setlength{\doublerulesep}{0.5pt}

  \begin{tabularx}{\linewidth}{
    >{\centering\arraybackslash}m{0.22\linewidth}||
    Y|Y|Y
    }
     & \textbf{Affine}                                                                                              & \textbf{One-Edge} & \textbf{$\kappa$-missing} \\
    \hhline{=||===}

    \textsc{Utilitarian Welfare}
     & P [\Cref{thm:affine_utilitarian}]
     & NP-complete [\Cref{thm:approval_utilitarian}]
     & \cellcolor[HTML]{FFFFCC}{}                                                                                                                                   \\
    \hhline{-||--}

    \textsc{Egalitarian Welfare}
     & \multicolumn{2}{c|}{\centering NP-complete [\Cref{thm:affine_egalitarian,thm:1_edge_big_kappa_egalitarian}]}
     & \cellcolor[HTML]{FFFFCC}{}                                                                                                                                   \\
    \hhline{-||--}

    \textsc{wPO-Construction}
     & \multicolumn{2}{c|}{\centering P [\Cref{thm:construct_wPO}]}
     & \cellcolor[HTML]{FFFFCC}{}                                                                                                                                   \\
    \hhline{-||--}

    \textsc{sPO-Construction}
     & P [\Cref{thm:affine_construct_sPO}]
     & NP-hard [\Cref{thm:1-edge_big_kappa_construct_sPO}]
     & \cellcolor[HTML]{FFFFCC}{}                                                                                                                                   \\
    \hhline{-||--}

    \textsc{wPO-Verification}
     & \multirow{2}{*}{\makecell[c]{coNP-complete                                                                                                                   \\[\smallskipamount][\Cref{thm:check_wsPO}]}}
     & \cellcolor[HTML]{CCFFCC}{P [\Cref{thm:1-edge_check_wPO_maximal}]}
     & \cellcolor[HTML]{FFFFCC}{}                                                                                                                                   \\
    \hhline{-||~-}
    \textsc{sPO-Verification}
     &
     & coNP-complete [\Cref{thm:check_wsPO}]
     & \cellcolor[HTML]{FFFFCC}{}                                                                                                                                   \\
    \hhline{-||--}

    \textsc{wCW-Existence}
     & \multirow{2}{*}{\makecell[c]{NP-hard                                                                                                                         \\[\smallskipamount][\Cref{thm:affine_exists_wsCW}]}}
     & P [\Cref{thm:approval_exists_wCW}]
     & \cellcolor[HTML]{FFFFCC}{}                                                                                                                                   \\
    \hhline{-||~-}
    \textsc{sCW-Existence}
     &
     & NP-hard [\Cref{thm:approval_exists_sCW}]
     & \cellcolor[HTML]{FFFFCC}{}                                                                                                                                   \\
    \hhline{-||--}

    \textsc{wCW-Verification}
     & \multicolumn{2}{c|}{\multirow{2}{*}{\makecell[c]{coNP-complete                                                                                               \\[\smallskipamount][\Cref{thm:check_wsCW}]}}}
     & \cellcolor[HTML]{FFFFCC}{}                                                                                                                                   \\
    \hhline{-||~~}
    \textsc{sCW-Verification}
     & \multicolumn{2}{c|}{}
     & \multirow[t]{-10}{*}{\cellcolor[HTML]{FFFFCC}{\makecell[c]{P (for constant $\kappa$)                                                                         \\[\smallskipamount][\Cref{thm:kappa_maximal_poly}]}}} \\
  \end{tabularx}
  \caption{Overview of the results for maximal matchings, with differences to the general case marked in green and yellow.}\label{tab:maximal_overview}
\end{table}

\section{Conclusion}
\label{sec:conclusion}
We proposed and analyzed a model for single-winner voting on matchings. Unlike prior work, our model captures a broad class of collective decision-making problems in which multiple stakeholders hold global preferences over entire matchings, rather than local preferences concerning their own assignments.
Our comprehensive complexity taxonomy reveals sharp boundaries between tractable and intractable cases. While some variants, such as \textsc{Utilitarian Welfare} under affine utilities or \textsc{Egalitarian Welfare} under $0$-missing approval and $1$-missing approval, remain computationally feasible, most others quickly become intractable. This demonstrates that the structure imposed by the matching constraint does not compensate for the exponential size of the outcome space, highlighting a fundamental limitation of applying classical single-winner aggregation principles in domains with complex candidate constraints.  Further restricting the candidate space, for instance, to maximal matchings, does not generally restore tractability. Designing tractable yet meaningful solution concepts for single-winner voting with complex, exponentially-sized candidate spaces remains an important direction for future research.

Further, it would be interesting to identify natural restrictions on the candidate space or on voter preferences that yield additional tractable cases. In particular, the parameterized complexity with respect to the number of voters remains an open question. For instance, can \textsc{Egalitarian Welfare} under affine utilities be solved in polynomial time for a constant number of voters?

While our analysis focuses on notions of social desirability, it would also be promising to study the complexity of winner determination for specific voting rules in the context of voting on matchings. Assuming that voters provide weak rankings of matchings by their overlap with their preferred matching, our results already answer several such questions:
For Borda, polynomial-time solvability follows  from the polynomial-time solvability of \textsc{Utilitarian Welfare} under affine utilities (see Footnote~\ref{fn:borda_plurality}).
For Plurality, winner determination corresponds to \textsc{Utilitarian Welfare} under $0$-missing approval, implying  NP-hardness in the general case and polynomial-time solvability under maximal matchings.
For Condorcet-consistent voting rules, the exact complexity remains open, but our hardness results for Condorcet existence and verification under affine utilities suggest computational intractability.

\clearpage
\appendix

\section{Relationship to Approval-Based Multiwinner Voting} \label{app:RW}
We now discuss in more detail the relationship between single-winner voting on matchings and approval-based multiwinner voting, where given a set of candidates, a committee size $k$, and voters' approval preferences over these candidates, a size-$k$ subset of these candidates needs to be selected.
Under affine utilities (with $a_v=1$ and $b_v=0$) for our setting, both settings are intuitively closely related: in the matching setting, the utility a voter has for a matching is the size of the overlap between the selected matching and the voter's preferred matching; in approval-based multiwinner voting, the utility a voter has for a committee is the size of the overlap between the committee and the voter's approval ballot.
The main distinction arises from the differing feasibility: cardinality constraints in multiwinner voting and matching constraints in our setting.

Comparing the complexity landscape between approval-based multiwinner voting and voting on matchings with affine utilities, we obtain a very similar picture. In both settings, maximizing utilitarian welfare is solvable in polynomial time, whereas maximizing egalitarian welfare, as well as verifying Pareto optimality and Condorcet winners, becomes (co)NP-complete \citep{aziz2020computing,darmann2013hard}.\footnote{In approval-based multiwinner voting, maximizing utilitarian welfare is achieved by standard approval voting and the hardness of maximizing egalitarian welfare can be shown by a simple reduction from \textsc{Vertex Cover}.} In fact, as we show in \Cref{thm:check_wsPO}, there is a simple reduction from verifying Pareto optimality in the multiwinner setting to the corresponding problem for voting on matchings under affine utilities.
However, this and similar straightforward reduction ideas stop working when restricting the candidate space to maximal matchings, as approval sets can no longer be easily translated into preferred matchings.
More generally speaking, a generic reduction between the two settings, independent of the considered computational problem, remains elusive.

\section{Full Proofs}
In the following, we provide full proofs for all statements from the main part of the paper. For hardness results under affine and one-edge approval utilities, we always restrict the preferred matchings and input matchings to maximal matchings, effectively showing that hardness even holds if the candidate space is restricted to maximal matchings.

\subsection{Social Welfare}

\affineegalitarian*
\begin{proof}
  The problem is in NP, as a candidate that provides a utility of at least $k$ for every voter can be verified in polynomial time. For NP-hardness, we reduce from \textsc{$3$-Sat}, which is famously NP-complete (see, for example, \citep{computers-intractability-garey}).

  \begin{definition}[\textsc{$3$-Sat}]
    \label{def:$3$-Sat}
    Let $X = \{x_1, \dots, x_n\}$ be a set of $n$ variables. For a variable $x_i$, we call $x_i$ and $\overline{x_i}$ literals. Given a truth assignment for the variable, the literal $x_i$ is \textit{true} if and only if the variable is \textit{true}, and $\overline{x_i}$ is \textit{true} otherwise. Let $C = \{c_1, \dots, c_m\}$ be a set of clauses, each containing three literals over $X$. A clause is satisfied if at least one of its literals is \textit{true}. The goal is to determine whether there exists an assignment of truth values to the variables in $X$ that simultaneously satisfies all clauses in $C$.
  \end{definition}

  Construction:
  For a variable $x_i$, let $C_i = \{c \in C \mid x_i \in c \vee \overline{x_i} \in c\}$ be the set of clauses that contain literals from $x_i$, and let $\overline{C_i} = C \setminus C_i$.
  We construct a graph $G$ consisting of gadgets for the variables. For a variable $x_i \in X$, we construct a gadget forming a star with $(n-2) \cdot |\overline{C_i}| + 2$ edges. We call the first two edges $e_i$ and $\overline{e_i}$. The remaining edges are labeled $(e_i)^a_b$ for $a \in [n-2]$ and $b \in [|\overline{C_i}|]$.
  For each clause $c \in C$, we construct $(n-2)$ voters, so $V$ contains $(n-2) \cdot m$ voters in total. For all voters $v \in V$, we set $\alpha_v = 1$ and $\beta_v = 0$ as coefficients in their utility functions. If a clause contains the literal $x_i$, all voters for that clause approve edge $e_i$; if it contains $\overline{x_i}$, they approve $\overline{e_i}$.
  As each clause contains exactly three literals from three distinct variables\footnote{A clause that contains both literals of a single variable is trivially satisfied and can be excluded from the instance.}, there are $n-3$ variables not involved in the clause. For each of these, each voter approves exactly one of the remaining edges. Formally, assume clause $c$ is the $j$th clause not containing variable $x_i$. Then the $(n-2)$ voters for $c$ approve the edges $(e_i)^1_j$ to $(e_i)^{n-2}_j$.

  The input for the \textsc{Egalitarian Welfare} instance consists of graph $G$, voters $V$, and $k = 1$.
  Note that all preferred matchings are maximal matchings, as each voter approves exactly one edge in each gadget.

  \begin{figure}[t]
    \centering
    \begin{tikzpicture}[node distance={20mm}, thick, main/.style = {draw, circle}]

    \node[main] (0) {};
    \node[main] (1) [above left of=0] {};
    \node[main] (2) [above right of=0] {};

    \node[main] (3) [angle={210}{of 0}] {};
    \node[main, draw=none] (4) [angle={230}{of 0}] {$...$};
    \node[main] (5) [angle={250}{of 0}] {};
    \node[main, draw=none] (6) [angle={270}{of 0}] {$...$};
    \node[main] (7) [angle={290}{of 0}] {};
    \node[main, draw=none] (8) [angle={310}{of 0}] {$...$};
    \node[main] (9) [angle={330}{of 0}] {};

    \draw[-] (0) edge node[midway, fill=white] {$e_i$} (1);
    \draw[-] (0) edge node[midway, fill=white] {$\overline{e_i}$} (2);
    \draw[-] (0) edge node[pos=.6, font=\footnotesize, fill=white] {${(e_i)^1_1}$} (3);
    \draw[-] (0) edge node[pos=.6, font=\footnotesize, fill=white] {$(e_i)^{n-2}_1$} (5);
    \draw[-] (0) edge node[pos=.6, font=\footnotesize, fill=white] {$(e_i)^1_{|\overline{C_i}|}$} (7);
    \draw[-] (0) edge node[pos=.6, font=\footnotesize, fill=white] {$(e_i)^{n-2}_{|\overline{C_i}|}$} (9);

\end{tikzpicture}
    \caption{Gadget for variable $x_i$.}
    \label{fig:affine_egalitarian_gadget}
  \end{figure}
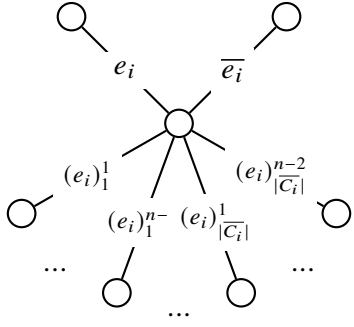

  Correctness:
  $[\Rightarrow]$
  Let $(X, C)$ be a Yes-instance for \textsc{$3$-Sat}. Then there is an assignment to $X$ that satisfies all clauses in $C$. Let $X^t$ be the variables assigned \textit{true}, and $X^f$ those assigned \textit{false}. Define the matching
  \(
  M := \{e_i \mid x_i \in X^t\} \cup \{\overline{e_i} \mid x_i \in X^f\}.
  \)
  This matching satisfies the matching property and is maximal, as exactly one edge from each gadget is selected. It provides utility of at least one to every voter, since each clause is satisfied, and each clause has at least one corresponding edge in $M$ that gives positive utility to all voters for that clause.

  $[\Leftarrow]$
  Assume there is a matching $M$ that provides utility of at least one to every voter. For a voter to have positive utility in $M$, at least one of their approved edges must be included. Consider an arbitrary clause $c$.
  It is not possible to satisfy all $n-2$ voters corresponding to clause $c$ using only edges not associated with literals in $c$: only $n-3$ such gadgets are available, and only one edge per gadget can be in the matching. Thus, at most $n-3$ voters can be satisfied this way. Since all $n-2$ voters have positive utility, at least one edge corresponding to a literal in $c$ must be present in $M$.
  Assigning \textit{true} to all variables $x_i$ with $e_i \in M$ and \textit{false} to those with $\overline{e_i} \in M$ yields a satisfying assignment for the \textsc{$3$-Sat} instance.
\end{proof}

\smallegalitarian*
\begin{proof}
  The proof for $\kappa = 1$ can be found in the main part of the paper. Here we show the case for $\kappa = 0$.

  Algorithm:
  Let $G = (W,E)$ be a graph and $V$ be the set of voters.
  We consider the set $\bigcup_{v \in V} M^*(v)$ and check whether it satisfies the matching constraint. This can be verified in $\mathcal{O}(|E|)$. If the constraint is satisfied, we return Yes; otherwise, No.

  Correctness:
  A voter approves a matching if and only if it includes their entire preferred matching. Therefore, a matching that is approved by all voters must be a superset of $\bigcup_{v \in V} M^*(v)$. If this union already violates the matching constraint, then no such matching can exist. Otherwise, $\bigcup_{v \in V} M^*(v)$ itself is a feasible matching with egalitarian welfare equal to one.
\end{proof}

\oneedgebigegalitarian*
\begin{proof}
  The problem is in NP, as a matching that is approved by every voter serves as a witness for a Yes-instance and can be verified in polynomial time.

  \paragraph{One-edge approval} We use exactly the same reduction as in \Cref{thm:affine_egalitarian}. There each voter was required to approve at least one edge, which is exactly the condition we have for approval under one-edge utilities.

  \paragraph{$\kappa$-missing approval} We reduce from \textsc{$\kappa+1$-Sat}. The idea is to use a similar construction as in \Cref{thm:affine_egalitarian}. Note that we do not require the preferred matchings to be maximal.

  Construction:
  Given a \textsc{$\kappa+1$-Sat} instance with variables $X$ and clauses $C$, we construct gadgets for each variable and one voter per clause. For each variable $x_i$, we construct a gadget consisting of a path with two edges, denoted $e_i$ and $\overline{e_i}$ (see \Cref{fig:affine_egalitarian_gadget_big_kappa}). For each clause $c_j$, we construct a voter $v_j$ who approves $e_i$ if $x_i \in c_j$, and approves $\overline{e_i}$ if $\overline{x_i} \in c_j$.

  \begin{figure}[tbp]
    \centering
    \begin{tikzpicture}[node distance={20mm}, thick, main/.style = {draw, circle}]

    \node[main] (0) {};
    \node[main] (1) [above left of=0] {};
    \node[main] (2) [above right of=0] {};

    \draw[-] (0) edge node[midway, fill=white] {$e_i$} (1);
    \draw[-] (0) edge node[midway, fill=white] {$\overline{e_i}$} (2);

\end{tikzpicture}
    \caption{Gadget for variable $x_i$.}
    \label{fig:affine_egalitarian_gadget_big_kappa}
  \end{figure}
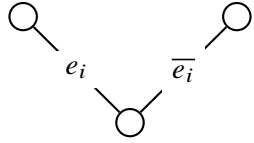

  Correctness:
  $[\Rightarrow]$
  Let $(X, C)$ be a Yes-instance for \textsc{$\kappa+1$-Sat}. Then there exists an assignment to the variables such that all clauses are satisfied. Let $X^t$ be the variables assigned \textit{true}, and $X^f$ the variables assigned \textit{false}. The edge set
  \(M := \{e_i \mid x_i \in X^t\} \cup \{\overline{e_i} \mid x_i \in X^f\}\)
  satisfies the matching property, since exactly one edge from each variable gadget is included. Each voter approves at least one edge in $M$, because they correspond to a clause that contains at least one \textit{true} literal. Since each voter's preferred matching has size $\kappa+1$, the inclusion of just one approved edge suffices for the voter to approve the matching. Hence, the matching has egalitarian welfare equal to one.

  $[\Leftarrow]$
  Assume there exists a matching $M$ with an egalitarian welfare of one. Then each voter approves at least one edge in $M$. Assigning \textit{true} to every variable $x_i$ for which $e_i \in M$, and \textit{false} otherwise, yields a satisfying assignment for the \textsc{$\kappa+1$-Sat} instance.
\end{proof}

\approvalutilitarian*
\begin{proof}
  The problem is in NP, as a matching that provides a utilitarian welfare of at least $k$ serves as a witness for a Yes-instance and can be verified in polynomial time.

  \paragraph{One-edge approval and $\kappa$-missing approval with $\kappa > 1$}
  We reduce from \textsc{Egalitarian Welfare}, which is NP-complete as shown in \Cref{thm:1_edge_big_kappa_egalitarian}.

  Construction:
  Given an instance of \textsc{Egalitarian Welfare}, defined by a graph $G$ and a set of voters $V$, we define the instance for \textsc{Utilitarian Welfare} using the same graph $G$, the same voters $V$, and set the target value $k = |V|$.

  Correctness:
  Since each voter contributes exactly one unit of utility if and only if they approve the matching, a utilitarian welfare of $|V|$ implies that all voters approve the matching. Hence, such a matching also has egalitarian welfare equal to one.

  \paragraph{$\kappa$-missing approval with $\kappa \leq 1$}
  We reduce from \textsc{Independent Set} (see, for example, \citep{computers-intractability-garey}).

  \begin{definition}[\textsc{Independent Set}]~\\
    Given a graph $G = (W,E)$ and a number $\mu \leq |W|$, the task is to decide whether there exists a set $I \subseteq W$ with $|I| \geq \mu$ such that for all edges $\{i,j\} \in E$, it holds that $i \notin I$ or $j \notin I$.
  \end{definition}
  Construction for $\kappa = 0$:
  Given an \textsc{Independent Set} instance with graph $\tilde G = (\tilde W, \tilde E)$ and $\mu \leq |\tilde W|$, construct a graph $G$ where for each edge $\epsilon = \{i,j\} \in \tilde E$ we construct a gadget consisting of three nodes and two edges, $e_\epsilon^i$ and $e_\epsilon^j$ (see \Cref{fig:Plurality_gadget}).
  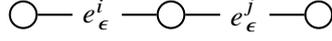
\begin{figure}[tbp]
    \centering
    \begin{tikzpicture}[node distance={20mm}, thick, main/.style = {draw, circle}]

    \node[main] (0) {};
    \node[main] (1) [left of=0] {};
    \node[main] (2) [right of=0] {};

    \draw[-] (0) edge node[midway, fill=white] {$e_\epsilon^i$} (1);
    \draw[-] (0) edge node[midway, fill=white] {$e_\epsilon^j$} (2);

\end{tikzpicture}
    \caption{Gadget for edge $\epsilon = \{i,j\}$.}
    \label{fig:Plurality_gadget}
  \end{figure}
  We construct $|V|$ voters $V$, one for each node in $\tilde W$. For a node $i \in \tilde W$, voter $v_i$ approves all edges $e_\epsilon^i$ for every $\epsilon \in E$.
  The \textsc{Utilitarian Welfare} instance is defined by graph $G$, voters $V$ and $k = \mu$.

  Construction for $\kappa = 1$:
  We begin with the same base construction as in the $\kappa = 0$ case and additionally, add a control gadget to $G$ consisting of a path of three edges $e_c^1, e_c^2$, and $e_c^3$ (see \Cref{fig:Plurality_gadget_control}). We then add $|V| + 1$ additional control voters $V^*$ who approve $e_c^1$ and $e_c^3$. For each voter in $V$ we add $e_c^2$ to the preferred matching.
  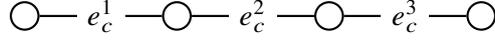
\begin{figure}[tbp]
    \centering
    \begin{tikzpicture}[node distance={20mm}, thick, main/.style = {draw, circle}]

    \node[main] (0) {};
    \node[main] (1) [right of=0] {};
    \node[main] (2) [right of=1] {};
    \node[main] (3) [right of=2] {};

    \draw[-] (0) edge node[midway, fill=white] {$e_c^1$} (1);
    \draw[-] (1) edge node[midway, fill=white] {$e_c^2$} (2);
    \draw[-] (2) edge node[midway, fill=white] {$e_c^3$} (3);

\end{tikzpicture}
    \caption{Control gadget.}
    \label{fig:Plurality_gadget_control}
  \end{figure}
  The \textsc{Utilitarian Welfare} instance is defined by graph $G$, voters $V \cup V^*$ and $k = |V|+1+\mu$.

  Correctness:
  For $\kappa = 1$, note that any matching including edge $e_c^2$ blocks both edges approved by the control voters. Thus, to exceed utility $|V|$, a matching must exclude $e_c^2$ and gain utility from all control voters. Consequently, each of the original voters must have all remaining edges of their preferred matchings included, which is equivalent to the $\kappa = 0$ case.

  $[\Rightarrow]$
  Suppose $(\tilde G, \mu)$ is a Yes-instance of \textsc{Independent Set}. Then there exists an independent set $I$ of at least $\mu$ mutually non-adjacent nodes. Define
  \(M := \{ e_\epsilon^i \mid \epsilon \in \tilde E,~ i \in I \text{ and } i \text{ is an endpoint of } \epsilon \}.\)
  This is a valid matching and is approved by each voter corresponding to a node in $I$. Hence, it achieves utilitarian welfare at least $\mu = k$.

  $[\Leftarrow]$
  Assume there is a matching $M$ that is approved by at least $k = \mu$ voters. Let $I := \{i \mid v_i \text{ approves } M \}$. Since the matching constraint prevents two edges from the same gadget (corresponding to edge $\epsilon = \{i,j\}$) from being included, no two nodes in $I$ are adjacent. Thus, $I$ is an independent set of size at least $\mu$.
\end{proof}

\subsection{Pareto Optimality}
\lemmascoreequivalence*
\begin{proof}
  Let $v \in V$ be a voter and let $u$ and $u'$ be two affine utility functions. Let $M, M'$ be two matchings. We aim to show that $u(v, M) > u(v, M')$ if and only if $u'(v, M) > u'(v, M')$, and $u(v, M)=u(v, M')$ if and only if $u'(v, M)=u'(v, M')$.
  Assume $|M \cap M^*(v)|=|M' \cap M^*(v)|$. By the definition of affine utility functions, it follows that $u(v, M) = u(v, M')$ and $u'(v, M) = u'(v, M')$.
  Now assume that $|M \cap M^*(v)| > |M' \cap M^*(v)|$. Since $\alpha_v > 0$ for all affine utility functions, it follows that $u(v, M) > u(v, M')$ and $u'(v, M) > u'(v, M')$.

  Therefore, the relative ranking of the matchings for a given voter depends only on the size of the overlap, not on the specific coefficients in the utility functions.
\end{proof}
\verificationhardness*
\begin{proof}
  Both problems are in coNP, as a matching that weakly or strongly Pareto dominates $M$ serves as a witness for a No-instance and can be verified in polynomial time.

  \paragraph{Affine Utilities} We show coNP-hardness by reducing from the verification of Pareto optimality in the multiwinner setting which was shown to be coNP-complete by \citet{aziz2020computing}. \citet{aziz2020computing} show coNP-completeness for verifying strong Pareto optimality in a preference model where agents only approve or disapprove candidates and strictly prefer a committee which contains more approved candidates.

  First, we extend their result to verifying weak Pareto optimality and then we provide reductions to \textsc{wPO-Verification} and \textsc{sPO-Verification} under affine utilities.

  For the coNP-hardness of verifying weak Pareto optimality in the multiwinner setting we reduce from \textsc{Vertex Cover}. Let $(\tilde G = (\tilde W, \tilde E), k)$ be an instance for \textsc{Vertex Cover}. We build the following instance for verifying Pareto optimality:
  \begin{itemize}
    \item The candidates are $A = \tilde W \cup D$, with $D = \{d_1, \dots d_k\}$
    \item The voters are the edges $N = \tilde E$, such that for an edge $e = \{u,w\}$ the corresponding voter approves candidates $\{u,w\}$.
    \item As input committee we choose $D$.
  \end{itemize}
  It is easy to verify that $D$ is not weakly Pareto optimal if and only if there exists a vertex cover for $G$ of size at most $k$, as the committee consisting of the corresponding candidates is approved by all voters, which is a Pareto improvement over $D$.

  \medskip
  We now reduce the verification of weak Pareto optimality in the multiwinner setting to \textsc{wPO-Verification} under affine utilities.

  Construction:
  Let $(N,A,D)$ be an instance for weak Pareto optimality in the multiwinner setting. We construct a \textsc{wPO-Verification} instance $(G, V, M)$ as follows:
  \begin{itemize}
    \item For each alternative $a \in A$ we construct a gadget consisting of two edges forming a path. We call the edges $e_a$ and $\overline e_a$.
    \item We add one additional gadget consisting of a single edge $e^*$.
    \item For each agent $i \in N$ we construct a voter $v_i$ that has $e_a$ in their preferred matching if and only if $a$ approved by agent $i$.
    \item We add two additional voter $v_1^*$ and $v_2^*$. Voter $v_1^*$ approves edges $\overline e_a$ for all $a \in A$ while $v_2^*$ all edges $e_a$. Both have and edge $e^*$ in their preferred matching.
    \item For all voters $v$ we choose $\alpha_v = 1$ and $\beta_v = 0$.
  \end{itemize}
  We choose matching $M = \{e_a ~|~ a \in D\} \cup \{\overline e_a ~|~ a \in A \setminus D\}$ as input matching.

  Correctness:
  $[\Rightarrow]$ Let $(N,A,D)$ be a No-instance for weak Pareto optimality in the multiwinner setting. Then there is a committee $D'$ of size $k$ such that every agent in $N$ prefers $D'$ over $D$. The matching $M' = \{e_a ~|~ a \in D'\} \cup \{\overline e_a ~|~ a \in A \setminus D'\} \cup \{e^*\}$ is strictly preferred by all voters in $V$ and $M$ is not weakly Pareto optimal.

  $[\Leftarrow]$ Let $(G, V, M)$ be a no instance for \textsc{wPO-Verification}. That is, there is a matching $M'$ that is strictly preferred by all voters. For a matching to be strictly preferred by all voters, voter $v_1^*$ must have a utility of at least $|A|-k+1$ and $v_2^*$ must have a utility of at least $k+1$. Thus, for exactly $k$ gadgets the edge $e_a$ must be selected. When choosing the set $\{a ~|~ e_a \in M'\}$ we get a committee of size $k$ which is strictly preferred over $D$ by all agents in $N$.

  \medskip
  For the coNP-completeness of \textsc{sPO-Verification} we reduce from verifying strong Pareto optimality in the multiwinner setting. The reduction is almost identical to the reduction from verifying weak Pareto optimality to \textsc{wPO-Verification}. The only difference is that we do not add edge $e^*$. The correctness holds analogously.

  \medskip
  Note that both reductions do not restrict the candidate space on maximal matchings. We proof the hardness for maximal matchings separately by reducing from the respective cases with arbitrary matchings.

  For \textsc{wPO-Verification} with maximal matchings we reduce from \textsc{wPO-Verification} with arbitrary matchings and create new edges that are approved by the voters in $V$ to create maximal preferred matchings. We duplicate the voters to prevent them from using the newly added edges to improve their utility. To make the input matching $M$ a maximal matching, we add additional edges to the graph that are approved by new voters corresponding to the nodes in the original graph.

  Construction:
  Given a \textsc{wPO-Verification} instance with arbitrary matchings, defined by $(G = (W, E), V, M)$, we construct the following \textsc{wPO-Verification} instance with maximal matchings.
  We construct a graph $\tilde G$ by adding edges to $G$.
  \begin{itemize}
    \item For each node $u \in W$, let $V_u$ be the set of voters that approve any edge incident to $u$. We add $|V \setminus V_u| \cdot (|W|+4) + |W| + 3$ new edges incident only to node $u$:
          \begin{itemize}
            \item Let the first three edges be $(e^*_u)_1$, $(e^*_u)_2$, and $(e^*_u)_3$.
            \item Let the next $|W|$ edges be $(\tilde e_u)_1, \dots, (\tilde e_u)_{|W|}$.
            \item Let the remaining edges be $(e_u)^a_b$ for $a \in [|W|+4]$ and $b \in [|V \setminus V_u|]$.
          \end{itemize}
    \item We add three additional star gadgets, each with $|V| \cdot (|W|+4) + 2$ edges and disjoint from the rest of the graph. For $k \in [3]$:
          \begin{itemize}
            \item Let the first two edges in the $k$th star be $(e_{S_k}^*)_1$ and $(e_{S_k}^*)_2$.
            \item Let the remaining edges be $(e_{S_k})^a_b$ for $a \in [|W|+4]$ and $b \in [|V|]$.
          \end{itemize}
  \end{itemize}
  We now construct the voter set $\tilde{V}$.
  For each voter $v \in V$, we create $|W|+4$ voters named $v^1, \dots, v^{|W|+4}$.
  Additionally, we create $|W|$ voters~$\tilde v_1, \dots, \tilde v_{|W|}$ and three voters $v_s^1, v_s^2, v_s^3$.
  Their preferences are defined as follows:
  \begin{itemize}
    \item For a voter $v \in V$ and $a \in [|W|+4]$ the voter $v^a$ approves $M^*(v)$. Let $u \in W$ such that $v \notin V_u$. Voter $v^a$ approves one of the edges that were added to $u$. Formally, let $v$ be the $b$th voter not in which does not approve any edge incident to $u$. Then voter $v^a$ approves edge $(e_u)^a_b$.
    \item For a voter $v \in V$, $a \in [|W|+4]$ and $k \in [3]$ a voter $v^a$ corresponding to a voter $v$ in the original graph approves one edge in the $k$th star gadget. Formally let $v$ be the $k$th original voter, then $v^a$ approves edge $(e_{S_k})^a_b$.
    \item For $i \in [|W|]$ and $u \in W$ the voter $\tilde{v_i}$ approves edges $(\tilde e_u)_i$. In the star gadgets the voter approves $(e_{S_k}^*)_1$ for $k \in [3]$.
    \item For each node $u$ and $k \in [3]$ the voter $v_s^k$ approves edge $(e^*_u)_k$. In the star gadget $v_s^k$ approves $(e_{S_k}^*)_2$ and $(e_{S_l}^*)_1$ for $l \in [3] \setminus \{k\}$.
  \end{itemize}
  The transformed input matching $\tilde M$ is defined as follows:
  \begin{itemize}
    \item It includes all edges in $M$.
    \item For every node $i \in W$ not incident to an edge in $M$, include edge $(\tilde e_i)_i$.
    \item For each star gadget $k \in [3]$, include edge $(e_{S_k}^*)_2$.
  \end{itemize}
  Note that all preferred matchings and the input matching $\tilde M$ are maximal matchings.
  The transformed instance is $(\tilde G, \tilde{V}, \tilde M)$.

  Correctness:
  $[\Rightarrow]$
  Suppose $(G, V, M)$ is a No-instance for \textsc{wPO-Verification}. Then, there exists a matching $M'$ that strongly Pareto dominates $M$. Define
  \(
  \tilde M' := M' \cup \{(e_{S_k}^*)_1 \mid k \in [3]\}.
  \)
  The utilities behave as follows:
  \begin{itemize}
    \item Each $v^a \in \tilde{V}$ has the same utility as the corresponding $v \in V$. Since $v$ is strictly better off in $M'$, so is $v^a$ in $\tilde M'$.
    \item Each $\tilde v_i$ has utility at most one in $\tilde M$, but utility $3$ in $\tilde M'$ due to the three edges in the star gadgets.
    \item Each $v_s^k$ has utility one in $\tilde M$ and utility $2$ in $\tilde M'$.
  \end{itemize}
  Hence, $\tilde M'$ strictly improves the utility of all voters in $\tilde{V}$, and $(\tilde G, \tilde{V}, \tilde M)$ is a No-instance for \textsc{wPO-Verification} with maximal matchings.

  $[\Leftarrow]$
  Suppose $(\tilde G, \tilde{V}, \tilde M)$ is a No-instance for \textsc{wPO-Verification}. Then there exists a matching $\tilde M'$ where every voter in $\tilde{V}$ is strictly better off than in $\tilde M$.
  Since all edges in $\tilde G$ are incident to one of the $|W|$ original nodes from $G$ or one of the three central nodes of the star gadgets, any matching can contain at most $|W| + 3$ edges.

  Let $v \in V$, with corresponding voters $v^1, \dots, v^{|W|+4} \in \tilde{V}$. Since all $|W|+4$ voters strictly improve, and each additional edge (outside of $G$) is approved by at most one of them, it is impossible to satisfy all using only added edges. Therefore, the increased utility must come from the original subgraph $G$.

  As the preferences of $v^a$ in $G$ equal those of $v$, and they are all strictly better off in $\tilde M'$, it follows that $M'$ strictly improves the utility of all voters in $V$ compared to $M$. Hence, $(G, V, M)$ is a No-instance for \textsc{wPO-Verification}.

  In \Cref{fig:affine-wPO-max-example1,fig:affine-wPO-max-example2,fig:affine-wPO-max-example3}, we show an example of the reduction for a simple No-instance of \textsc{wPO-Verification}.

  \begin{figure}[tbp]
    \centering

    \begin{subfigure}{0.9\linewidth}
      \centering
      \begin{tikzpicture}[node distance={30mm}, thick, main/.style = {draw, circle}]

    \node[main] (0) {};
    \node[main] (1) [right of=0] {};
    \node[main] (2) [right of=1] {};
    \node[main] (3) [right of=2] {};
    \node[main] (4) [right of=3] {};
    \node[main] (5) [right of=4] {};

    \draw[-] (0) edge node[above, fill=white] {$\{v_1, v_2, v_3\}$} (1);
    \draw[-] (1) edge node[above, fill=white] {$\{v_4\}$} (2);
    \draw[-, red] (2) edge node[above, fill=white] {$\{v_1\}$} (3);
    \draw[-] (3) edge node[above, fill=white] {$\{v_2\}$} (4);
    \draw[-] (4) edge node[above, fill=white] {$\{v_4\}$} (5);

\end{tikzpicture}
      \caption{Instance with input graph $G$ and matching $M$.}
      \label{fig:affine-wPO-max-example1}
    \end{subfigure}

    \vspace{0.8em}

    \begin{subfigure}{0.9\linewidth}
      \centering
      \includestandalone[width=0.9\linewidth]{figures/affine_pareto_strong_transformed}
      \caption{Transformed graph $\tilde G$ with maximal preferred matchings and input matching $\tilde M$.}
      \label{fig:affine-wPO-max-example2}
    \end{subfigure}

    \vspace{0.8em}

    \begin{subfigure}{0.9\linewidth}
      \centering
      \includestandalone[width=0.9\linewidth]{figures/affine_pareto_strong_transformed1}
      \caption{Matching $\tilde M'$ that encodes a dominating matching for $\tilde M$.}
      \label{fig:affine-wPO-max-example3}
    \end{subfigure}

    \caption{Example for \textsc{wPO-Verification}.}
    \label{fig:affine-wPO-max-example}
  \end{figure}
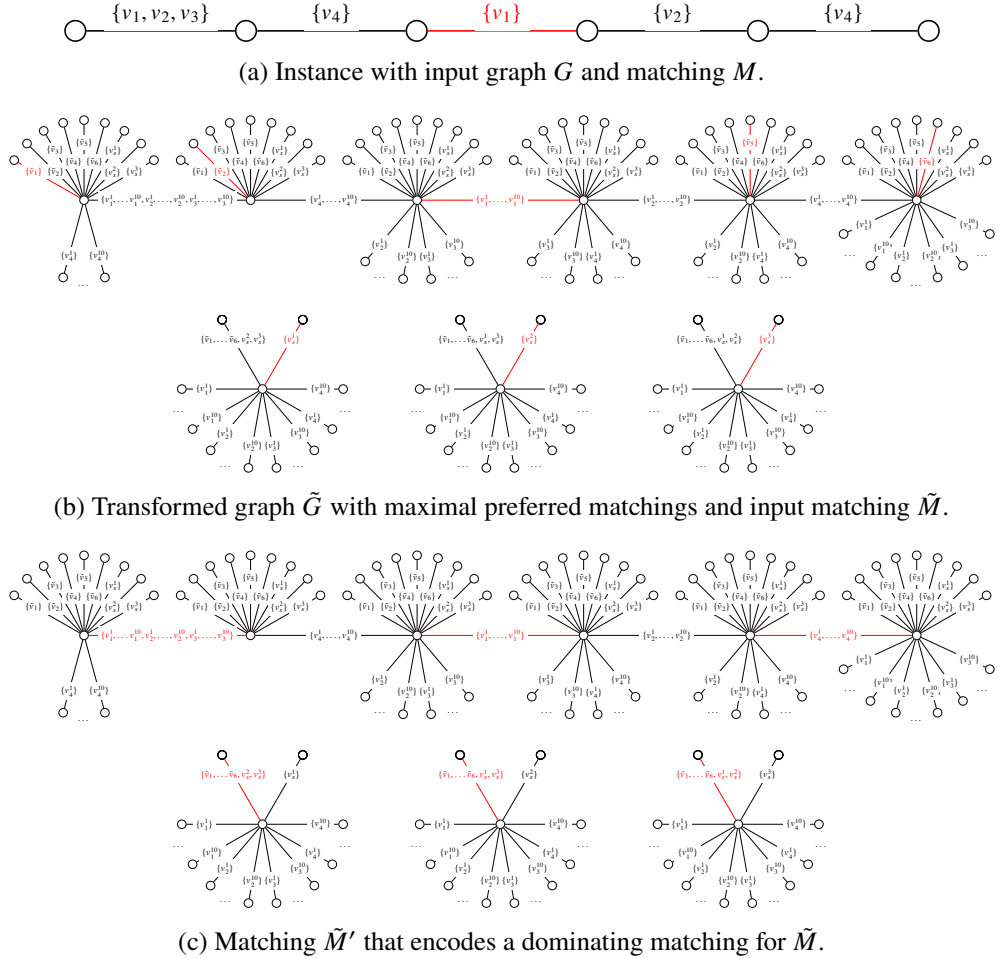

  \medskip
  For the coNP-hardness of \textsc{sPO-Verification} with maximal matchings we reduce from \textsc{sPO-Verification} with arbitrary matchings.
  The idea is to define additional edges that are approved by all voters in the original instance to create maximal preferred matchings. We then use new voters to block these edges from being in a matching that creates a Pareto improvement.

  First, we show that a matching $M$ that is not maximal can never be a strongly Pareto optimal matching. Assume $M$ is not maximal. Then we can find a strict superset $M'$ of $M$ that is still a matching. By definition, every edge is approved by at least one voter and utilities are strictly increasing which higher overlap. Therefore, there is at least one voter who assigns a higher utility to $M'$, and $M'$ weakly Pareto dominates $M$.

  Construction:
  Given an instance $(G, V, M)$ for \textsc{sPO-Verification} with arbitrary matchings, we construct an instance with maximal matchings. If the input matching $M$ is \textit{not} a maximal matching, then we return a trivial No-instance for \textsc{sPO-Verification} with maximal matchings.

  Let $M$ be a maximal matching. We create a graph $\tilde G$, by adding nodes and edges to $G$. For every node $u \in W$ in the original graph we add six nodes and seven edges which form a gadget as described in \Cref{fig:affine_sPO_max_gadget}. We also define the following edge sets:
  \begin{itemize}
    \item $E^{*1} =  \left \{e_u^{*1}   \mid u \in W \right \}$
    \item $E^{*2} =  \left \{e_u^{*2}   \mid u \in W \right \}$
    \item $E^{**1} = \left \{e_u^{**1}  \mid u \in W \right \}$
    \item $E^{**2} = \left \{e_u^{**2}  \mid u \in W \right \}$
  \end{itemize}
  The set of voters $\tilde{V}$ consists of the original voters $V$ and two additional voters $v_1^*$ and $v_2^*$. The preferred matchings are defined as follows:
  \begin{itemize}
    \item For a voter $v \in V$ the corresponding $\tilde v \in \tilde{V}$ approves all edges in $M^*(v)$. In addition the voter approves $\{e_u^2, e_u^3 \mid u \in W\}$. If $v$ does not approve any edge incident to node $u$, $\tilde v$ also approves edge $e_u^1$.
    \item Voter $v_1^*$ approves all edges in $E^{**1}$ and in $E^{*1}$ and voter $v_2^*$ approves all edges in $E^{**2}$ and in $E^{*2}$.
  \end{itemize}
  We define the input matching $\tilde M$ by adding $E^{**1}$ and $E^{**2}$ to $M$. An example for the transformation of the graph and input matching $\tilde M$ can be seen in \Cref{fig:affine_check_sPO_maximal_example1,fig:affine_check_sPO_maximal_example3}.
  The \textsc{sPO-Verification} with maximal matching is defined by $(\tilde G, \tilde{V}, \tilde M)$.
  \begin{figure}[tbp]
    \centering
    \begin{tikzpicture}[node distance={30mm}, thick, main/.style = {draw, circle}]

    \node[main] (0) {$u$};


    \node[main] (01) at ($(0)+(135:3cm)$)  {};
    \node[main] (02) at ($(0)+(90:3cm)$)  {};
    \node[main] (03) at ($(0)+(45:2cm)$)  {};

    \node[main] (04) at ($(02)+(90:2cm)$)  {};
    \node[main] (05) at ($(03)+(90:2cm)$)  {};
    \node[main] (06) at ($(03)+(0:2cm)$)  {};

    \draw[-] (0) edge node[midway, fill=white] {$e_u^1$} (01);
    \draw[-] (0) edge node[midway, fill=white] {$e_u^{*1}$} (02);
    \draw[-] (0) edge node[midway, fill=white] {$e_u^{*2}$} (03);

    \draw[-, red] (01) edge node[midway, fill=white] {$e_u^{**2}$} (02);

    \draw[-] (02) edge node[midway, fill=white] {$e_u^2$} (04);
    \draw[-] (03) edge node[midway, fill=white] {$e_u^3$} (05);
    \draw[-, red] (03) edge node[midway, fill=white] {$e_u^{**1}$} (06);

\end{tikzpicture}
    \caption{Gadget for node $u$.}
    \label{fig:affine_sPO_max_gadget}
  \end{figure}
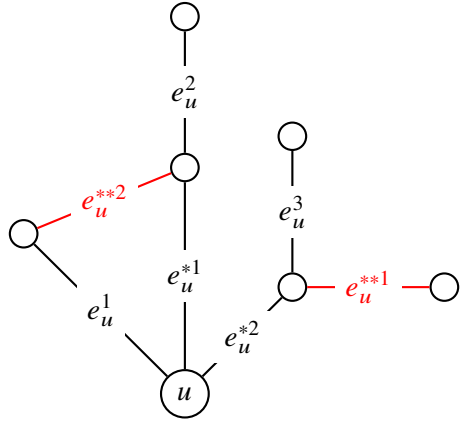

  Correctness:
  $[\Rightarrow]$
  Let $(G, V, M)$ be a No-instance for \textsc{sPO-Verification} with arbitrary matchings. We can assume that $M$ is a maximal matching, as otherwise a trivial No-instance is returned. There is a matching $M'$ in $G$ that weakly dominates matching $M$. Consider matching $\tilde M' := M' \cup E^{**1} \cup E^{**2}$ in $G'$. It is a matching, as $E^{**1} \cup E^{**2}$ are not adjacent to any nodes in the original graph. Voters $v_1^*$ and $v_2^*$ have the same utility for $\tilde M$ and $\tilde M'$. For all voters in $V$ it holds by the definition of $M'$ that they have at least the same utility as in $M$ and at least one voter has strictly higher utility. Thus, the corresponding voters in $\tilde{V}$ have increased utility in $\tilde M'$. Therefore, $\tilde M'$ weakly Pareto dominates $\tilde M$ and the transformed instance is a No-instance of \textsc{sPO-Verification} with maximal matchings.

  $[\Leftarrow]$
  Let $(\tilde G, \tilde{V}, \tilde M)$ be a No-instance for \textsc{wPO-Verification} with maximal matchings. Therefore, there is a matching $\tilde M'$ that weakly Pareto dominates $\tilde M$. We can note the following about $\tilde M'$.
  \begin{enumerate}
    \item For every node $u \in W$ at most two of the edges $\{e_u^{*1}, e_u^{*2}, e_u^{**1}, e_u^{**2}\}$ are in $\tilde M'$.
    \item For every node $u \in W$ exactly two of the edges $\{e_u^{*1}, e_u^{*2}, e_u^{**1}, e_u^{**2}\}$ are in $\tilde M'$.
    \item For every node $u \in W$ none of the edges $\{e_u^1, e_u^2, e_u^3\}$ are in $\tilde M'$.
    \item When restricting the transformed graph to the subgraph induced by the original graph $G$, matching $\tilde M'$ on that subgraph weakly Pareto dominates matching $M$ for voters $V$.
  \end{enumerate}
  We now prove these five statements.
  \begin{enumerate}
    \item Considering any subset set of three edges from $\{e_u^{*1}, e_u^{*2}, e_u^{**1}, e_u^{**2}\}$, two of them share at least one node. As this violates the matching constraint, at most two of those edges can be in any matching.
    \item Voters $v_1^*$ and $v_2^*$ each have a utility of $|V|$ in matching $\tilde M$. As matching $\tilde M'$ weakly Pareto dominates $\tilde M$ they both have utility of at least $|V|$ in $\tilde M'$. Using (1) we know that each gadget can contribute at most $2$ utility to the sum of their utilities. Since there are $|V|$ gadgets and the utilities must sum to at least $2|V|$, each gadget must contribute a utility of exactly $2$.
    \item Considering every set of two edges from $\{e_u^{*1}, e_u^{*2}, e_u^{**1}, e_u^{**2}\}$ those edges form a maximal matching in the gadget corresponding to $u$. Therefore no additional edge in that gadget can be added.
    \item From (1) and (2) we know that the utility that is contributed by the gadgets is $2|V|$. As voters $v_1^*$ and $v_2^*$ only get utility through the gadgets and both have utility if $|V|$ in $\tilde M$, they both have utility $|V|$ in $\tilde M'$ and do not improve. Using (3) we know that no utility for the voters correspinding to original voters can come from the gadgets. In $\tilde M'$ all original voters have at least the same utilities as in $\tilde M$ and at least one voter has strictly increased utility. That also holds if we restrict $\tilde M'$ to the subgraph induced by original graph $G$. Therefore, $\tilde M'$ on that subgraph weakly Pareto dominates matching $M$.
  \end{enumerate}
  From (4) it follows directly that $(G, V, M)$ is a No-instance for \textsc{sPO-Verification}. An example can be seen in \Cref{fig:affine_check_sPO_maximal_example1,fig:affine_check_sPO_maximal_example2,fig:affine_check_sPO_maximal_example3,fig:affine_check_sPO_maximal_example4}.

  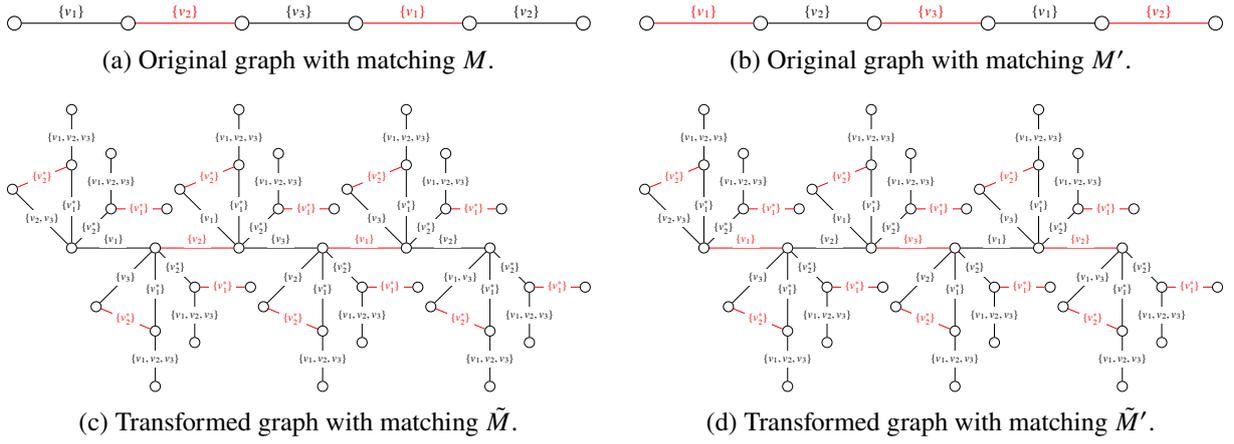
\begin{figure}[tbp]
    \centering

    \begin{subfigure}{0.48\linewidth}
      \centering
      \begin{tikzpicture}[node distance={30mm}, thick, main/.style = {draw, circle}]

    \node[main] (0) {};
    \node[main] (1) [right of=0] {};
    \node[main] (2) [right of=1] {};
    \node[main] (3) [right of=2] {};
    \node[main] (4) [right of=3] {};
    \node[main] (5) [right of=4] {};

    \draw[-] (0) edge node[above, fill=white] {$\{v_1\}$} (1);
    \draw[-, red] (1) edge node[above, fill=white] {$\{v_2\}$} (2);
    \draw[-] (2) edge node[above, fill=white] {$\{v_3\}$} (3);
    \draw[-, red] (3) edge node[above, fill=white] {$\{v_1\}$} (4);
    \draw[-] (4) edge node[above, fill=white] {$\{v_2\}$} (5);

\end{tikzpicture}
      \caption{Original graph with matching $M$.}
      \label{fig:affine_check_sPO_maximal_example1}
    \end{subfigure}\hfill
    \begin{subfigure}{0.48\linewidth}
      \centering
      \begin{tikzpicture}[node distance={30mm}, thick, main/.style = {draw, circle}]

    \node[main] (0) {};
    \node[main] (1) [right of=0] {};
    \node[main] (2) [right of=1] {};
    \node[main] (3) [right of=2] {};
    \node[main] (4) [right of=3] {};
    \node[main] (5) [right of=4] {};

    \draw[-, red] (0) edge node[above, fill=white] {$\{v_1\}$} (1);
    \draw[-] (1) edge node[above, fill=white] {$\{v_2\}$} (2);
    \draw[-, red] (2) edge node[above, fill=white] {$\{v_3\}$} (3);
    \draw[-] (3) edge node[above, fill=white] {$\{v_1\}$} (4);
    \draw[-, red] (4) edge node[above, fill=white] {$\{v_2\}$} (5);

\end{tikzpicture}
      \caption{Original graph with matching $M'$.}
      \label{fig:affine_check_sPO_maximal_example2}
    \end{subfigure}

    \vspace{0.8em}

    \begin{subfigure}{0.48\linewidth}
      \centering
      \begin{tikzpicture}[node distance={30mm}, thick, main/.style = {draw, circle}]

    \node[main] (0) {};
    \node[main] (1) [right of=0] {};
    \node[main] (2) [right of=1] {};
    \node[main] (3) [right of=2] {};
    \node[main] (4) [right of=3] {};
    \node[main] (5) [right of=4] {};

    \draw[-] (0) edge node[above, fill=white] {$\{v_1\}$} (1);
    \draw[-, red] (1) edge node[above, fill=white] {$\{v_2\}$} (2);
    \draw[-] (2) edge node[above, fill=white] {$\{v_3\}$} (3);
    \draw[-, red] (3) edge node[above, fill=white] {$\{v_1\}$} (4);
    \draw[-] (4) edge node[above, fill=white] {$\{v_2\}$} (5);


    \node[main] (01) at ($(0)+(135:3cm)$)  {};
    \node[main] (02) at ($(0)+(90:3cm)$)  {};
    \node[main] (03) at ($(0)+(45:2cm)$)  {};

    \node[main] (04) at ($(02)+(90:2cm)$)  {};
    \node[main] (05) at ($(03)+(90:2cm)$)  {};
    \node[main] (06) at ($(03)+(0:2cm)$)  {};

    \draw[-] (0) edge node[midway, fill=white] {$\{v_2,v_3\}$} (01);
    \draw[-] (0) edge node[midway, fill=white] {$\{v_1^*\}$} (02);
    \draw[-] (0) edge node[midway, fill=white] {$\{v_2^*\}$} (03);

    \draw[-, red] (01) edge node[midway, fill=white] {$\{v_2^*\}$} (02);

    \draw[-] (02) edge node[midway, fill=white] {$\{v_1, v_2, v_3\}$} (04);
    \draw[-] (03) edge node[midway, fill=white] {$\{v_1, v_2, v_3\}$} (05);
    \draw[-, red] (03) edge node[midway, fill=white] {$\{v_1^*\}$} (06);


    \node[main] (11) at ($(1)+(225:3cm)$)  {};
    \node[main] (12) at ($(1)+(270:3cm)$)  {};
    \node[main] (13) at ($(1)+(315:2cm)$)  {};

    \node[main] (14) at ($(12)+(270:2cm)$)  {};
    \node[main] (15) at ($(13)+(270:2cm)$)  {};
    \node[main] (16) at ($(13)+(0:2cm)$)  {};

    \draw[-] (1) edge node[midway, fill=white] {$\{v_3\}$} (11);
    \draw[-] (1) edge node[midway, fill=white] {$\{v_1^*\}$} (12);
    \draw[-] (1) edge node[midway, fill=white] {$\{v_2^*\}$} (13);

    \draw[-, red] (11) edge node[midway, fill=white] {$\{v_2^*\}$} (12);

    \draw[-] (12) edge node[midway, fill=white] {$\{v_1, v_2, v_3\}$} (14);
    \draw[-] (13) edge node[midway, fill=white] {$\{v_1, v_2, v_3\}$} (15);
    \draw[-, red] (13) edge node[midway, fill=white] {$\{v_1^*\}$} (16);


    \node[main] (21) at ($(2)+(135:3cm)$)  {};
    \node[main] (22) at ($(2)+(90:3cm)$)  {};
    \node[main] (23) at ($(2)+(45:2cm)$)  {};

    \node[main] (24) at ($(22)+(90:2cm)$)  {};
    \node[main] (25) at ($(23)+(90:2cm)$)  {};
    \node[main] (26) at ($(23)+(0:2cm)$)  {};

    \draw[-] (2) edge node[midway, fill=white] {$\{v_1\}$} (21);
    \draw[-] (2) edge node[midway, fill=white] {$\{v_1^*\}$} (22);
    \draw[-] (2) edge node[midway, fill=white] {$\{v_2^*\}$} (23);

    \draw[-, red] (21) edge node[midway, fill=white] {$\{v_2^*\}$} (22);

    \draw[-] (22) edge node[midway, fill=white] {$\{v_1, v_2, v_3\}$} (24);
    \draw[-] (23) edge node[midway, fill=white] {$\{v_1, v_2, v_3\}$} (25);
    \draw[-, red] (23) edge node[midway, fill=white] {$\{v_1^*\}$} (26);


    \node[main] (31) at ($(3)+(225:3cm)$)  {};
    \node[main] (32) at ($(3)+(270:3cm)$)  {};
    \node[main] (33) at ($(3)+(315:2cm)$)  {};

    \node[main] (34) at ($(32)+(270:2cm)$)  {};
    \node[main] (35) at ($(33)+(270:2cm)$)  {};
    \node[main] (36) at ($(33)+(0:2cm)$)  {};

    \draw[-] (3) edge node[midway, fill=white] {$\{v_2\}$} (31);
    \draw[-] (3) edge node[midway, fill=white] {$\{v_1^*\}$} (32);
    \draw[-] (3) edge node[midway, fill=white] {$\{v_2^*\}$} (33);

    \draw[-, red] (31) edge node[midway, fill=white] {$\{v_2^*\}$} (32);

    \draw[-] (32) edge node[midway, fill=white] {$\{v_1, v_2, v_3\}$} (34);
    \draw[-] (33) edge node[midway, fill=white] {$\{v_1, v_2, v_3\}$} (35);
    \draw[-, red] (33) edge node[midway, fill=white] {$\{v_1^*\}$} (36);


    \node[main] (41) at ($(4)+(135:3cm)$)  {};
    \node[main] (42) at ($(4)+(90:3cm)$)  {};
    \node[main] (43) at ($(4)+(45:2cm)$)  {};

    \node[main] (44) at ($(42)+(90:2cm)$)  {};
    \node[main] (45) at ($(43)+(90:2cm)$)  {};
    \node[main] (46) at ($(43)+(0:2cm)$)  {};

    \draw[-] (4) edge node[midway, fill=white] {$\{v_3\}$} (41);
    \draw[-] (4) edge node[midway, fill=white] {$\{v_1^*\}$} (42);
    \draw[-] (4) edge node[midway, fill=white] {$\{v_2^*\}$} (43);

    \draw[-, red] (41) edge node[midway, fill=white] {$\{v_2^*\}$} (42);

    \draw[-] (42) edge node[midway, fill=white] {$\{v_1, v_2, v_3\}$} (44);
    \draw[-] (43) edge node[midway, fill=white] {$\{v_1, v_2, v_3\}$} (45);
    \draw[-, red] (43) edge node[midway, fill=white] {$\{v_1^*\}$} (46);


    \node[main] (51) at ($(5)+(225:3cm)$)  {};
    \node[main] (52) at ($(5)+(270:3cm)$)  {};
    \node[main] (53) at ($(5)+(315:2cm)$)  {};

    \node[main] (54) at ($(52)+(270:2cm)$)  {};
    \node[main] (55) at ($(53)+(270:2cm)$)  {};
    \node[main] (56) at ($(53)+(0:2cm)$)  {};

    \draw[-] (5) edge node[midway, fill=white] {$\{v_1,v_3\}$} (51);
    \draw[-] (5) edge node[midway, fill=white] {$\{v_1^*\}$} (52);
    \draw[-] (5) edge node[midway, fill=white] {$\{v_2^*\}$} (53);

    \draw[-, red] (51) edge node[midway, fill=white] {$\{v_2^*\}$} (52);

    \draw[-] (52) edge node[midway, fill=white] {$\{v_1, v_2, v_3\}$} (54);
    \draw[-] (53) edge node[midway, fill=white] {$\{v_1, v_2, v_3\}$} (55);
    \draw[-, red] (53) edge node[midway, fill=white] {$\{v_1^*\}$} (56);

\end{tikzpicture}
      \caption{Transformed graph with matching $\tilde M$.}
      \label{fig:affine_check_sPO_maximal_example3}
    \end{subfigure}\hfill
    \begin{subfigure}{0.48\linewidth}
      \centering
      \begin{tikzpicture}[node distance={30mm}, thick, main/.style = {draw, circle}]

    \node[main] (0) {};
    \node[main] (1) [right of=0] {};
    \node[main] (2) [right of=1] {};
    \node[main] (3) [right of=2] {};
    \node[main] (4) [right of=3] {};
    \node[main] (5) [right of=4] {};

    \draw[-, red] (0) edge node[above, fill=white] {$\{v_1\}$} (1);
    \draw[-] (1) edge node[above, fill=white] {$\{v_2\}$} (2);
    \draw[-, red] (2) edge node[above, fill=white] {$\{v_3\}$} (3);
    \draw[-] (3) edge node[above, fill=white] {$\{v_1\}$} (4);
    \draw[-, red] (4) edge node[above, fill=white] {$\{v_2\}$} (5);


    \node[main] (01) at ($(0)+(135:3cm)$)  {};
    \node[main] (02) at ($(0)+(90:3cm)$)  {};
    \node[main] (03) at ($(0)+(45:2cm)$)  {};

    \node[main] (04) at ($(02)+(90:2cm)$)  {};
    \node[main] (05) at ($(03)+(90:2cm)$)  {};
    \node[main] (06) at ($(03)+(0:2cm)$)  {};

    \draw[-] (0) edge node[midway, fill=white] {$\{v_2,v_3\}$} (01);
    \draw[-] (0) edge node[midway, fill=white] {$\{v_1^*\}$} (02);
    \draw[-] (0) edge node[midway, fill=white] {$\{v_2^*\}$} (03);

    \draw[-, red] (01) edge node[midway, fill=white] {$\{v_2^*\}$} (02);

    \draw[-] (02) edge node[midway, fill=white] {$\{v_1, v_2, v_3\}$} (04);
    \draw[-] (03) edge node[midway, fill=white] {$\{v_1, v_2, v_3\}$} (05);
    \draw[-, red] (03) edge node[midway, fill=white] {$\{v_1^*\}$} (06);


    \node[main] (11) at ($(1)+(225:3cm)$)  {};
    \node[main] (12) at ($(1)+(270:3cm)$)  {};
    \node[main] (13) at ($(1)+(315:2cm)$)  {};

    \node[main] (14) at ($(12)+(270:2cm)$)  {};
    \node[main] (15) at ($(13)+(270:2cm)$)  {};
    \node[main] (16) at ($(13)+(0:2cm)$)  {};

    \draw[-] (1) edge node[midway, fill=white] {$\{v_3\}$} (11);
    \draw[-] (1) edge node[midway, fill=white] {$\{v_1^*\}$} (12);
    \draw[-] (1) edge node[midway, fill=white] {$\{v_2^*\}$} (13);

    \draw[-, red] (11) edge node[midway, fill=white] {$\{v_2^*\}$} (12);

    \draw[-] (12) edge node[midway, fill=white] {$\{v_1, v_2, v_3\}$} (14);
    \draw[-] (13) edge node[midway, fill=white] {$\{v_1, v_2, v_3\}$} (15);
    \draw[-, red] (13) edge node[midway, fill=white] {$\{v_1^*\}$} (16);


    \node[main] (21) at ($(2)+(135:3cm)$)  {};
    \node[main] (22) at ($(2)+(90:3cm)$)  {};
    \node[main] (23) at ($(2)+(45:2cm)$)  {};

    \node[main] (24) at ($(22)+(90:2cm)$)  {};
    \node[main] (25) at ($(23)+(90:2cm)$)  {};
    \node[main] (26) at ($(23)+(0:2cm)$)  {};

    \draw[-] (2) edge node[midway, fill=white] {$\{v_1\}$} (21);
    \draw[-] (2) edge node[midway, fill=white] {$\{v_1^*\}$} (22);
    \draw[-] (2) edge node[midway, fill=white] {$\{v_2^*\}$} (23);

    \draw[-, red] (21) edge node[midway, fill=white] {$\{v_2^*\}$} (22);

    \draw[-] (22) edge node[midway, fill=white] {$\{v_1, v_2, v_3\}$} (24);
    \draw[-] (23) edge node[midway, fill=white] {$\{v_1, v_2, v_3\}$} (25);
    \draw[-, red] (23) edge node[midway, fill=white] {$\{v_1^*\}$} (26);


    \node[main] (31) at ($(3)+(225:3cm)$)  {};
    \node[main] (32) at ($(3)+(270:3cm)$)  {};
    \node[main] (33) at ($(3)+(315:2cm)$)  {};

    \node[main] (34) at ($(32)+(270:2cm)$)  {};
    \node[main] (35) at ($(33)+(270:2cm)$)  {};
    \node[main] (36) at ($(33)+(0:2cm)$)  {};

    \draw[-] (3) edge node[midway, fill=white] {$\{v_2\}$} (31);
    \draw[-] (3) edge node[midway, fill=white] {$\{v_1^*\}$} (32);
    \draw[-] (3) edge node[midway, fill=white] {$\{v_2^*\}$} (33);

    \draw[-, red] (31) edge node[midway, fill=white] {$\{v_2^*\}$} (32);

    \draw[-] (32) edge node[midway, fill=white] {$\{v_1, v_2, v_3\}$} (34);
    \draw[-] (33) edge node[midway, fill=white] {$\{v_1, v_2, v_3\}$} (35);
    \draw[-, red] (33) edge node[midway, fill=white] {$\{v_1^*\}$} (36);


    \node[main] (41) at ($(4)+(135:3cm)$)  {};
    \node[main] (42) at ($(4)+(90:3cm)$)  {};
    \node[main] (43) at ($(4)+(45:2cm)$)  {};

    \node[main] (44) at ($(42)+(90:2cm)$)  {};
    \node[main] (45) at ($(43)+(90:2cm)$)  {};
    \node[main] (46) at ($(43)+(0:2cm)$)  {};

    \draw[-] (4) edge node[midway, fill=white] {$\{v_3\}$} (41);
    \draw[-] (4) edge node[midway, fill=white] {$\{v_1^*\}$} (42);
    \draw[-] (4) edge node[midway, fill=white] {$\{v_2^*\}$} (43);

    \draw[-, red] (41) edge node[midway, fill=white] {$\{v_2^*\}$} (42);

    \draw[-] (42) edge node[midway, fill=white] {$\{v_1, v_2, v_3\}$} (44);
    \draw[-] (43) edge node[midway, fill=white] {$\{v_1, v_2, v_3\}$} (45);
    \draw[-, red] (43) edge node[midway, fill=white] {$\{v_1^*\}$} (46);


    \node[main] (51) at ($(5)+(225:3cm)$)  {};
    \node[main] (52) at ($(5)+(270:3cm)$)  {};
    \node[main] (53) at ($(5)+(315:2cm)$)  {};

    \node[main] (54) at ($(52)+(270:2cm)$)  {};
    \node[main] (55) at ($(53)+(270:2cm)$)  {};
    \node[main] (56) at ($(53)+(0:2cm)$)  {};

    \draw[-] (5) edge node[midway, fill=white] {$\{v_1,v_3\}$} (51);
    \draw[-] (5) edge node[midway, fill=white] {$\{v_1^*\}$} (52);
    \draw[-] (5) edge node[midway, fill=white] {$\{v_2^*\}$} (53);

    \draw[-, red] (51) edge node[midway, fill=white] {$\{v_2^*\}$} (52);

    \draw[-] (52) edge node[midway, fill=white] {$\{v_1, v_2, v_3\}$} (54);
    \draw[-] (53) edge node[midway, fill=white] {$\{v_1, v_2, v_3\}$} (55);
    \draw[-, red] (53) edge node[midway, fill=white] {$\{v_1^*\}$} (56);

\end{tikzpicture}
      \caption{Transformed graph with matching $\tilde M'$.}
      \label{fig:affine_check_sPO_maximal_example4}
    \end{subfigure}

    \caption{Example for \textsc{sPO-Verification}.}
    \label{fig:affine_check_sPO_maximal_example}
  \end{figure}

  \paragraph{Approval-based Utilities}
  When considering \textsc{wPO-Verification} under one-edge approval and $\kappa$-missing approval with $\kappa > 1$ we reduce from the corresponding \textsc{Egalitarian Welfare} problem, which was shown to be NP-complete in \Cref{thm:1_edge_big_kappa_egalitarian}.

  Construction:
  Let $G$ be a graph and $V$ a set of voters. The instance for \textsc{wPO-Verification} is defined by the graph $G$, voters $V$, and $M = \emptyset$.\footnote{Note that in this theorem we do not require the matchings to be maximal. In \Cref{thm:1-edge_check_wPO_maximal} we show that \textsc{wPO-Verification} under one-edge approval is trivial when restricting to maximal matchings.}

  Correctness:
  $[\Rightarrow]$
  Assume $(G, V)$ is a Yes-instance for \textsc{Egalitarian Welfare}. Then there is a matching $M'$ that is approved by all voters in $V$. This matching strongly Pareto dominates $M = \emptyset$.

  $[\Leftarrow]$ Assume there is a matching $M'$ that strongly Pareto dominates $M = \emptyset$. In this matching, all voters must have strictly higher utility than in $M$. As all voters have utility zero in $M$, they all have utility one in $M'$, so $(G, V)$ is a Yes-instance for \textsc{Egalitarian Welfare}.

  \medskip
  For \textsc{sPO-Verification} under one-edge approval we also reduce from the respective \textsc{Egalitarian Welfare} problems. The idea is to create an additional edge and additional voters such that all voters except one approve a matching $M$ in the new graph. The one remaining voter blocks the edge responsible for the approval from the original voters. Only when the instance is a Yes-instance for \textsc{Egalitarian Welfare} does there exist a matching that weakly Pareto dominates $M$. We use that we showed that \textsc{Egalitarian Welfare} is NP-complete even on graphs that are collections of stars.

  Construction:
  Let $G$ be a graph that is a collection of stars, and let $V$ be a set of voters. We divide $G$ into $m$ connected components $S_1, \dots, S_m$ such that each component is a star.

  We construct a graph $G'$ that includes $G$ as a subgraph and adds the following new nodes and edges:
  \begin{itemize}
    \item A control gadget consisting of two edges forming a path, denoted $e_1$ and $e_2$.
    \item For each star component $S_i$, we add $m + 1$ new edges connected to the center of the component, denoted $e^i_1$ through $e^i_{m+1}$.
  \end{itemize}

  Next, we construct $m + 1$ new voters $V^* := \{v^*_1, \dots, v^*_{m+1}\}$. For each $j \in [m+1]$, the preferred matching of voter $v^*_j$ is $M^*(v^*_j) := \{e^i_j \mid i \in [m]\} \cup \{e_2\}$.
  For each original voter $v \in V$, we add edge $e_1$ to their preferred matching.

  Note that all preferred matchings are still maximal matchings if they were maximal before. Let $M := \{e_1\} \cup \{e^i_i \mid i \in [m]\}$. This matching is approved by all voters except voter $v^*_{m+1}$.
  The \textsc{sPO-Verification} instance is defined by graph $G'$, voters $V \cup V^*$, and matching $M$.

  Correctness:
  $[\Rightarrow]$ Assume $(G, V)$ is a Yes-instance for \textsc{Egalitarian Welfare}. That is, there is a matching $M'$ in $G$ that is approved by all voters in $V$. Consider matching $M' \cup \{e_2\}$ in $G'$. This matching is approved by all voters in $V$ and all voters in $V^*$, and therefore weakly Pareto dominates $M$.

  $[\Leftarrow]$ Now assume there is a matching $M'$ in $G'$ that weakly Pareto dominates $M$. This matching must be approved by all voters, since all voters except $v^*_{m+1}$ already approve $M$.

  Note that for each of the $m$ star components, at most one edge can be included in $M'$. Since no two voters in $V^*$ approve the same edge in the star components, at most one of them can be satisfied using an edge from one star. Thus, the only way for a matching to be approved by all voters in $V^*$ is to include edge $e_2$. Consequently, edge $e_1$ is not in $M'$. All other edges approved by voters from $V$ are in the original subgraph $G$. Hence, matching $M'$ restricted to $G$ is approved by all voters in $V$ and is therefore a matching with egalitarian welfare one.

  \medskip
  For \textsc{sPO-Verification} under $\kappa$-missing approval with $\kappa > 1$ we  reduce from \textsc{$\kappa+1$-Sat}. The idea is to create voters corresponding to the clauses of the \textsc{$\kappa+1$-Sat} instance. Using enough control gadgets, we can ensure that each clause-voter requires only one edge from the variable gadgets to approve the matching. Additional control voters prevent a variable from being simultaneously assigned \textit{true} and \textit{false}. Our proof differs from \Cref{thm:1_edge_big_kappa_egalitarian} in that we need to construct a matching approved by all but one voter.

  Construction:
  Consider an \textsc{$\kappa+1$-Sat} instance with variables $X = \{x_1, \dots, x_n\}$ and clauses $C = \{c_1, \dots c_m\}$.
  We now construct an instance for \textsc{sPO-Verification}.

  The graph $G$ consists of the following components:
  \begin{itemize}
    \item We construct one start gadget with two edges forming a path. We call the two edges $e_s^1$ and $e_s^2$ (see \Cref{fig:subfigure:wPO_approaval_start}).
    \item For each variable $x_i \in X$ we construct a gadget consisting of a circle with six edges. We call these edges $e_i^1$ to $e_i^6$ (see \Cref{fig:subfigure:wPO_approaval_variable}).
    \item We then construct $2 \kappa$ control gadgets. The first $\kappa$ of them are stars consisting of three edges while the last $\kappa$ of them are paths of length two. We call the edges $\hat e_{j}^1$ to $\hat  e_{j}^3$ for $j \in [\kappa]$ and $\hat e_{j}^1$ to $\hat e_{j}^2$ for $j \in [\kappa+1, 2\kappa]$ (see \Cref{fig:subfigure:wPO_approaval_control}).
  \end{itemize}

  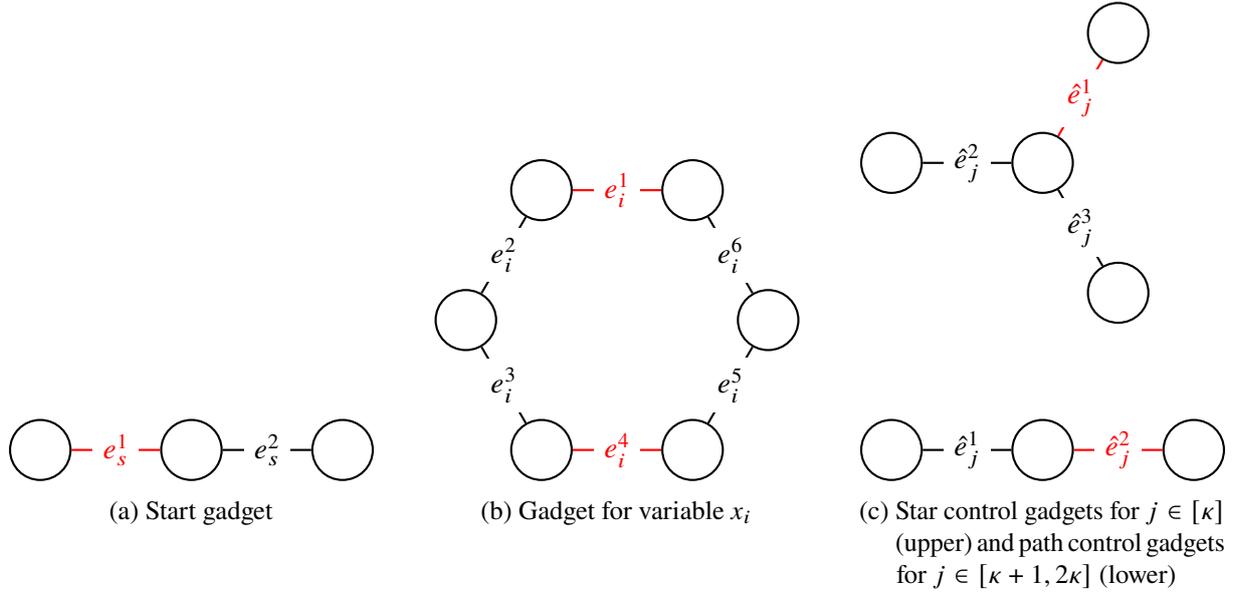
\begin{figure}[tbp]
    \centering

    \begin{subfigure}[t]{0.3\linewidth}
      \centering
      \begin{tikzpicture}[thick, node distance=20mm, main/.style = {draw, circle, minimum size=8mm}]
  \node[main] (s1) {};
  \node[main, right of=s1] (s2) {};
  \draw[red] (s1) -- node[midway, fill=white] {$e_s^1$} (s2);

  \node[main, right of=s2] (s3) {};
  \draw (s2) -- node[midway, fill=white] {$e_s^2$} (s3);
\end{tikzpicture}
      \caption{Start gadget}
      \label{fig:subfigure:wPO_approaval_start}
    \end{subfigure}\hfill
    \begin{subfigure}[t]{0.3\linewidth}
      \centering
      \begin{tikzpicture}[thick, main/.style = {draw, circle, minimum size=8mm}]
  \foreach \i in {1,...,6} {
      \node[main] (v\i) at (60*\i:2cm) {};
    }

  \draw[red] (v1) -- node[midway, fill=white] {$e_i^1$} (v2);
  \draw (v2) -- node[midway, fill=white] {$e_i^2$} (v3);
  \draw (v3) -- node[midway, fill=white] {$e_i^3$} (v4);
  \draw[red] (v4) -- node[midway, fill=white] {$e_i^4$} (v5);
  \draw (v5) -- node[midway, fill=white] {$e_i^5$} (v6);
  \draw (v6) -- node[midway, fill=white] {$e_i^6$} (v1);
\end{tikzpicture}
      \caption{Gadget for variable $x_i$}
      \label{fig:subfigure:wPO_approaval_variable}
    \end{subfigure}\hfill
    \begin{subfigure}[t]{0.3\linewidth}
      \centering
      \begin{tikzpicture}[thick, node distance=20mm, main/.style = {draw, circle, minimum size=8mm}]

  \node[main] (c1) {};

  \node[main] (c1a) at (60:2cm) {};
  \node[main] (c1b) at (180:2cm) {};
  \node[main] (c1c) at (300:2cm) {};

  \draw[red] (c1) -- node[midway, fill=white] {$\hat e_j^1$} (c1a);
  \draw (c1) -- node[midway, fill=white] {$\hat e_j^2$} (c1b);
  \draw (c1) -- node[midway, fill=white] {$\hat e_j^3$} (c1c);

  \node[main, below=3cm of c1] (c2b) {};
  \node[main, left of=c2b] (c2a) {};
  \node[main, right of=c2b] (c2c) {};
  \draw (c2a) -- node[midway, fill=white] {$\hat e_{j}^1$} (c2b);
  \draw[red] (c2b) -- node[midway, fill=white] {$\hat e_{j}^2$} (c2c);

\end{tikzpicture}
      \caption{Star control gadgets for $j\in[\kappa]$ (upper) and path control gadgets for $j\in[\kappa+1,2\kappa]$ (lower)}
      \label{fig:subfigure:wPO_approaval_control}
    \end{subfigure}

    \caption{Gadgets used in the reduction for \textsc{wPO-Verification} under approval utilities.}
    \label{fig:wPO_approval_gadgets}
  \end{figure}

  We construct $|C| + 2|X| + 3$ voters. Their approvals are defined as follows:
  \begin{itemize}
    \item For each clause $c_j \in C$ we have a voter $v_j$. Voter $v_j$ approves $e_i^1$ if $x_i \in c_j$ and approves $e_i^4$ if $\overline{x_i} \in c_j$.
    \item For each variable $x_i$ we have two voters $v_i^a$ and $v_i^b$. Voter $v_i^a$ approves $e_i^2$ and $e_i^5$ while $v_i^b$ approves $e_i^3$ and $e_i^6$. In addition both voters approve edge $e_s^1$ and edges $\hat e_{j}^3$ for $j \in [\kappa - 2]$.
    \item We have an additional voter $v^*$. This voter approves $e_s^2$ and $\hat e_{j}^3$ for $j \in [\kappa]$.
    \item We then have two voters $v_1^*$ and $v_2^*$. Voter $v_1^*$ approves $\hat e_{j}^1$ for $j \in [2\kappa]$ and $v_2^*$ approves $\hat e_{j}^2$ for $j \in [2\kappa]$.
  \end{itemize}
  We define our input matching as
  \(M := \left \{e_i^1, e_i^4 \mid x_i \in X \right \} \cup \left \{e_s^1 \right \} \cup \left \{\hat e_j^1 \mid j \in [\kappa] \right \} \cup \left \{\hat e_{j}^2 \mid j \in [\kappa+1, 2\kappa] \right \}.\)

  Note that this input matching is approved by all voters except voter $v^*$.

  The idea of this construction is that edges $e_i^1$ and $e_i^4$ represent the literals for variable $x_i$. A voter corresponding to a clause is satisfied if at least one edge corresponding to a literal in their clause is in the matching. In the input matching all voters corresponding to clauses start with utility one. The only voter that can improve is the start voter $v^*$ which initiates a chain reaction that removes utility from voters in the variables. To regain utility these voters have to block one of the edges corresponding to the literals, which makes sure that a variable cannot simultaneously be assigned \textit{true} and \textit{false}.

  Correctness:
  $[\Rightarrow]$
  Assume $(X, C)$ is a Yes-Instance for \textsc{$\kappa+1$-Sat}. Therefore there exists a satisfying assignment for the variables. Consider the following matching:
  \begin{align*}
    M' := ~ & \left \{e_i^1, e_i^3, e_i^5 \mid x_i \in X \text{ is \textit{true}} \right \} \cup \left \{e_i^2, e_i^4, e_i^6 \mid x_i \in X \text{ is \textit{false}} \right \} \\
            & \cup \left \{e_s^2 \right \} \cup \left \{\hat e_j^1 \mid j \in [\kappa] \right \} \cup \left \{\hat e_{j}^2 \mid j \in [\kappa+1, 2\kappa] \right \}.
  \end{align*}
  We can show the following statements about matching $M'$:
  \begin{enumerate}
    \item For all $c_j \in C$ matching $M'$ is approved by voter $v_j$.
    \item For all $x_i \in X$ matching $M'$ is approved by $v_i^a$ and $v_i^b$.
    \item Matching $M'$ is approved by $v^*, v_1^*$ and $v_2^*$.
  \end{enumerate}
  We now prove these statements:
  \begin{enumerate}
    \item Let $c_j \in C$ be a clause. This clause has $\kappa + 1$ literals and therefore voter $v_j$ approves $\kappa + 1$ edges. We know that in the satisfying assignment at least one of the literals in $c_j$ is \textit{true}.  Therefore the intersection of $M'$ and $M^*(v_j)$ must be at least one and at most $\kappa$ edges in $M^*(v_j)$ are not in $M$.
    \item Let $x_i \in X$ be a variable. Voters $v_i^a$ and $v_i^b$ approve $\kappa + 1$ edges each. Thus, at least one of the edges in their preferred matching must be in $M'$. If $x_i$ is \textit{true},  this are edges $e_i^3$ for $v_i^b$ and $e_i^5$ for $v_i^a$ and If $x_i$ is \textit{false}, this are edges $e_i^2$ for $v_i^a$ and $e_i^6$ for $v_i^b$.
    \item Voter $v^*$ approves $\kappa + 1$ edges and therefore at least one edge from $M^*(v^*)$ must be in $M'$. It holds that $M' \cap M^*(v^*) = \{e_s^2\}$. Voter $v_1^*$ approves $2 \kappa$ edges and the $\kappa$ edges $\{\hat e_j^1 \mid j \in [\kappa]\}$ are in $M'$. Voter $v_2^*$ also approves $2 \kappa$ edges and the $\kappa$ edges $\{\hat e_{j}^2 \mid j \in [\kappa+1, 2\kappa]\}$ are in $M'$.
  \end{enumerate}
  As matching $M'$ is approved by all voters, it weakly Pareto dominates matching $M$.

  $[\Leftarrow]$
  Assume $M$ is not strongly Pareto optimal. That is, there is a matching $M'$ that weakly Pareto dominates $M$. As $M$ is approved by all but one voter, this matching must be approved by all voters.
  We can show the following about $M'$:
  \begin{enumerate}
    \item For each $j \in \kappa$ the edge $\hat e_j^3 \notin M'$
    \item Edge $e_s^2$ is in $M'$.
    \item For each variable $x_i \in X:$ either $\{e_i^2,e_i^6\} \in M'$ or $\{e_i^5, e_i^3\} \in M'$.
    \item For each variable $x_i \in X: e_i^1 \notin M'$ or $e_i^4 \notin M'$.
    \item The assignment $x_i = \text{true}$ if and only if $e_i^1 \in M'$ is a satisfying assignment.
  \end{enumerate}
  We now show these statements.
  \begin{enumerate}
    \item Consider voters $v_1^*$ and $v_2^*$. They both approve $2 \kappa$ edges. For them to approve a matching, this matching must contain at least $\kappa$ of those edges. Note that $v_1^*$ and $v_2^*$ never approve the same an edge. Also note that they only approve edges in the control gadgets and that each control gadget only provides one utility to at most one of them. As they must have a combined utility of $2 \kappa$ and there are only $2 \kappa$ control gadgets, for half of the control gadgets edge $\hat e_{j}^1$ is in $M'$ and fot the other half $\hat e_{j}^2$ is in $M'$. Therefore there is no $j \in \kappa$ such that $\hat e_{j}^3$ is in $M'$.
    \item We know that matching $M'$ is approved by all voters. As voter $v^*$ approves $\kappa + 1$ edges and by (1) we know that the edges $\hat e_{j}^3$ for $j \in [\kappa]$ are not in $M'$, the remaining edge, that is edge $e_s^2$, is in $M'$.
    \item As by (2) edge $e_s^2$ is in $M'$, we know that $e_s^1$ is not in $M'$. For voters $v_i^a$ and $v_i^b$ corresponding to variable $x_i$ to approve $M'$, $M'$ must contain at leat one of the edges from their preferred matching. By (1) we know that the $\kappa - 2$ edges $\hat e_{j}^3$ for $j \in [\kappa - 2]$ are not in $M'$. Therefore either $e_i^2$ and $e_i^6$ or $e_i^5$ and $e_i^3$ are in $M'$.
    \item From statement (3) we know that for each variable $i \in [n]$ either $e_i^2$ and $e_i^6$ or $e_i^5$ and $e_i^3$ are in $M'$. As $e_i^2$ and $e_i^6$ are adjacent to $e_i^1$ and $e_i^5$ and $e_i^3$ are adjacent to $e_i^4$, only one of those edges is in $M'$.
    \item We know that for each clause $c_j \in C$ voter $v_j$ approves matching $M'$. That means that $M'$ shares at least one edge with their preferred matching. From (4) we know that for each variable either $e_i^1 \notin M'$ or $e_i^4 \notin M'$. Therefore each variable has an unambiguous assignment to \textit{true} or \textit{false}.
  \end{enumerate}
  From (5) follows that there is a satisfying assignment for the \textsc{$\kappa+1$-Sat} instance.
\end{proof}

\approvalconstructspo*
\begin{proof}
  In both cases, we reduce from \textsc{Egalitarian Welfare}\footnote{Unlike in the rest of this paper, we need to use a Turing reduction in this proof, as we reduce from a decision problem to a search problem. This allows us to use \textsc{sPO-Construction} as a subroutine for solving \textsc{Egalitarian Welfare}. Note that the definition of NP-hardness under Turing reductions slightly differs from the definition of NP-hardness under Karp reductions. For example, under Turing reductions there is no distinction between NP-hardness and coNP-hardness, while there is a difference under Karp reductions.}, which was shown to be NP-complete in \Cref{thm:1_edge_big_kappa_egalitarian}. The idea is to use the fact that when matchings with egalitarian welfare of one exist, such matchings are the only strongly Pareto optimal ones.

  Construction:
  Let $G$ be a graph an $V$ be a set of voters. Using \textsc{sPO-Construction} with the same graph and the same voters we get a strongly Pareto optimal matching on $G$. If the matching is approved by all voters in $V$, we return that a matching with egalitarian welfare of one exists. Otherwise we return, that there is no matching with positive egalitarian welfare.

  Correctness:
  A matching that is approved by all voters is always strongly Pareto optimal, as there is no possibility for any voter to improve their utility. It weakly dominates any matching that is not approved by all voters. Thus, if there is a matching with positive egalitarian welfare, only matchings with positive egalitarian welfare are strongly Pareto optimal and \textsc{sPO-Construction} outputs such a matching.
\end{proof}

\subsection{Condorcet Winners}

\cwverification*
\begin{proof}
  The problem lies in coNP, since a matching $M'$ that wins or ties in a pairwise comparison against $M$ serves as a witness for a No-instance.

  \paragraph{Affine Utilities}
  For \textsc{wCW-Verification} under affine utilities we reduce from \textsc{wPO-Verification} , which was shown to be coNP-hard in \Cref{thm:check_wsPO}.

  Construction:
  Consider an instance for \textsc{wPO-Verification}, given by graph $G$, matching $M$, and voters $V$.
  To construct an instance for \textsc{wCW-Verification}, we define $|V| - 1$ additional voters $V' = \{v'_1, \dots, v'_{|V|-1}\}$. For each voter $v' \in V'$, we use $M^*(v') = M$ as their preferred matching.
  The resulting \textsc{wCW-Verification} instance consists of graph $G$, matching $M$, and voters $V \cup V'$.
  A complete example of the construction is shown in \Cref{fig:check_sCW_example1,fig:check_sCW_example2}.

  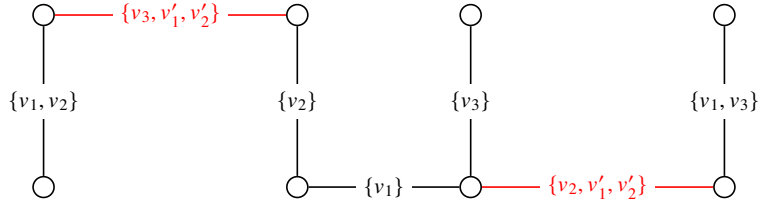
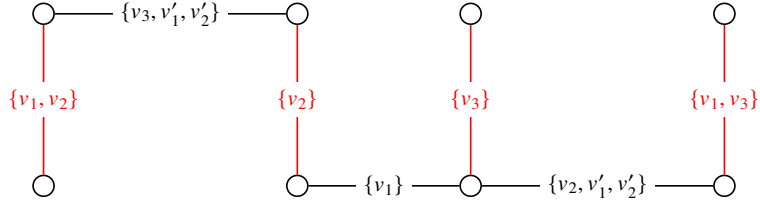
\begin{figure}[tbp]
    \centering

    \begin{subfigure}{0.9\linewidth}
      \centering
      \begin{tikzpicture}[node distance={30mm}, thick, main/.style = {draw, circle}]
    \node[main] (0) {};
    \node[main] (1) [above of=0] {};
    \node[main] (2) [right=4cm of 1] {};
    \node[main] (3) [below of=2] {};
    \node[main] (4) [right of=3] {};
    \node[main] (5) [above of=4] {};
    \node[main] (6) [right=4cm of 4] {};
    \node[main] (7) [above of=6] {};

    \draw[-] (0) edge node[midway, fill=white] {$\{v_1, v_2\}$} (1);
    \draw[-, red] (1) edge node[midway, fill=white] {$\{v_3, v'_1, v'_2\}$} (2);
    \draw[-] (2) edge node[midway, fill=white] {$\{v_2\}$}(3);
    \draw[-] (3) edge node[midway, fill=white] {$\{v_1\}$}(4);
    \draw[-] (4) edge node[midway, fill=white] {$\{v_3\}$}(5);
    \draw[-, red] (4) edge node[midway, fill=white] {$\{v_2, v'_1, v'_2\}$}(6);
    \draw[-] (6) edge node[midway, fill=white] {$\{v_1, v_3\}$}(7);

\end{tikzpicture}
      \caption{Input matching $M$. The new voters $v'_1$ and $v'_2$ have $M$ as their preferred matching.}
      \label{fig:check_sCW_example1}
    \end{subfigure}

    \vspace{0.8em}

    \begin{subfigure}{0.9\linewidth}
      \centering
      \begin{tikzpicture}[node distance={30mm}, thick, main/.style = {draw, circle}]
    \node[main] (0) {};
    \node[main] (1) [above of=0] {};
    \node[main] (2) [right=4cm of 1] {};
    \node[main] (3) [below of=2] {};
    \node[main] (4) [right of=3] {};
    \node[main] (5) [above of=4] {};
    \node[main] (6) [right=4cm of 4] {};
    \node[main] (7) [above of=6] {};

    \draw[-, red] (0) edge node[midway, fill=white] {$\{v_1, v_2\}$} (1);
    \draw[-] (1) edge node[midway, fill=white] {$\{v_3, v'_1, v'_2\}$} (2);
    \draw[-, red] (2) edge node[midway, fill=white] {$\{v_2\}$}(3);
    \draw[-] (3) edge node[midway, fill=white] {$\{v_1\}$}(4);
    \draw[-, red] (4) edge node[midway, fill=white] {$\{v_3\}$}(5);
    \draw[-] (4) edge node[midway, fill=white] {$\{v_2, v'_1, v'_2\}$}(6);
    \draw[-, red] (6) edge node[midway, fill=white] {$\{v_1, v_3\}$}(7);

\end{tikzpicture}
      \caption{Matching $M'$ preferred by all original voters.}
      \label{fig:check_sCW_example2}
    \end{subfigure}

    \caption{Example for \textsc{sCW-Verification}.}
    \label{fig:check_sCW_example}
  \end{figure}

  Correctness:
  $[\Rightarrow]$
  Let $(G, V, M)$ be a No-instance for \textsc{wPO-Verification}. Then, there exists a matching $M'$ s.th.\ all voters in $V$ strictly prefer $M'$ over $M$. In the corresponding \textsc{wCW-Verification} instance, $M'$ wins in a pairwise comparison with $M$: the $|V|$ original voters strictly prefer $M'$, while the $|V| - 1$ added voters strictly prefer $M$. Thus, $M$ is not a weak Condorcet winner.

  $[\Leftarrow]$
  Let $(G, V \cup V', M)$ be a No-instance for \textsc{wCW-Verification}. Then, there exists a matching $M'$ that wins a pairwise comparison against $M$. Since every added voter in $V'$ has maximum utility in $M$ and strictly lower utility in every other matching, they all strictly prefer $M$ over $M'$. There are $2|V| - 1$ voters in total, and at least $|V| - 1$ prefer $M$ over $M'$. Therefore, for $M'$ to win, all $|V|$ original voters must strictly prefer $M'$ over $M$. Hence, $M$ is strongly Pareto dominated by $M'$ and not weakly Pareto optimal.

  For \textsc{sCW-Verification} the proof is completely analogous to \textsc{wCW-Verification} with the difference that the number of additional voters is $|V|$ instead of $|V| - 1$. In every case in which we argue a matching wins against input matching $M$ it instead wins or ties.

  \paragraph{One-edge approval}
  For the hardness \textsc{wCW-Verification} under one-edge approval preferences we reduce from \textsc{$3$-Sat}.
  The idea is to construct gadgets, similar to the proof of \Cref{thm:1_edge_big_kappa_egalitarian}, that correspond to the literals of the variables. We then add voters for the clauses which approve the matching if the corresponding clause is fulfilled. Using control voters we construct a matching that is approved by all but one voter. This matching is a weak Condorcet winner unless there exists a fulfilling assignment for the \textsc{$3$-Sat} instance.

  Construction:
  Let $X = \{x_1, \dots, x_n\}$ be a set of variables and $C = \{c_1, \dots, c_m\}$ be a set of clauses. For a variable $x_i$, let $C_i = \{c \in C \mid x_i \in c \vee \overline{x_i} \in c\}$ be the set of clauses that contain literals from $x_i$ and $\overline{C_i} = C \setminus C_i$.
  We construct a graph $G$ consisting of the following gadgets (see \Cref{fig:check_wCW_one_edge_gadget}):
  \begin{itemize}
    \item For each variable $x_i \in W$ we construct a gadget forming a star with $2 + 2n - 1 + (n-2) \cdot |\overline{C_i}|$ edges. We label these edges as follows:
          \begin{itemize}
            \item Let the first two edges be $e_i$ and $\overline{e_i}$.
            \item Let the next $2n-1$ edges be ${(e_i^*)}_1$ to ${(e_i^*)}_{2n-1}$.
            \item Let the last $(n-2) \cdot |\overline{C_i}|$ edges be $(e_i)^a_b$ for $a \in [n-2]$ and $b \in [|\overline{C_i}|]$.
          \end{itemize}
    \item We construct one control gadget with three edges $e_C^1$, $e_C^2$, and $e_C^3$.
  \end{itemize}
  \begin{figure}[tbp]
    \centering
    \begin{tikzpicture}[node distance=20mm, thick, main/.style = {draw, circle, minimum size=7mm}]
    \node[main] (0) at (0,0) {};
    \node[main] (e1) at (120:2cm) {};
    \node[main] (e2) at (60:2cm) {};
    \draw (0) -- node[midway, fill=white] {$e_i$} (e1);
    \draw (0) -- node[midway, fill=white] {$\overline{e_i}$} (e2);
    \node[main] (a1) at (210:2.5cm) {};
    \node[main, draw=none] (a2) at (230:2.5cm) {$\dots$};
    \node[main] (a3) at (250:2.5cm) {};
    \draw (0) -- node[pos=0.7, fill=white, font=\footnotesize] {$(e_i^*)_1$} (a1);
    \draw (0) -- node[pos=0.7, fill=white, font=\footnotesize] {$(e_i^*)_{2n-1}$} (a3);
    \node[main] (b1) at (290:2.5cm) {};
    \node[main, draw=none] (b2) at (310:2.5cm) {$\dots$};
    \node[main] (b3) at (330:2.5cm) {};
    \draw (0) -- node[pos=0.7, fill=white, font=\footnotesize] {$(e_i)^1_1$} (b1);
    \draw (0) -- node[pos=0.6, fill=white, font=\footnotesize] {$(e_i)^{n-2}_{|\overline{C_i}|}$} (b3);
    \node[main] (c1) at (3, 0) {};
    \node[main, right=1.5cm of c1] (c2) {};
    \node[main, above=1.5cm of c2] (c3) {};
    \node[main, right=1.5cm of c2] (c4) {};

    \draw (c1) -- node[midway, fill=white] {$e_C^1$} (c2);
    \draw (c2) -- node[midway, fill=white] {$e_C^2$} (c3);
    \draw (c2) -- node[midway, fill=white] {$e_C^3$} (c4);

\end{tikzpicture}
    \caption{Gadget for variable $x_i$ and control gadget.}
    \label{fig:check_wCW_one_edge_gadget}
  \end{figure}
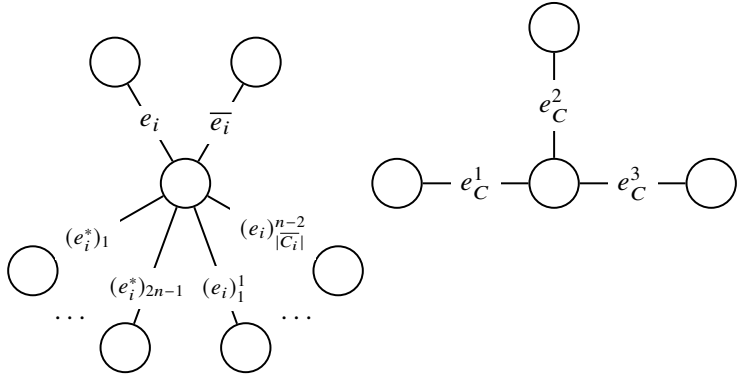
  We construct the following $(n-2)m + 2n$ voters with their preferred matchings:
  \begin{itemize}
    \item For each clause $c_j \in C$ we construct $(n-2)$ voters $V_j$. Let these voters be labeled $v_j^1$ to $v_j^{n+1}$. If $c_j$ contains literal $x_i$, all voters for that clause approve $e_i$ and if $c_j$ contains $\overline{x_i}$, the voters approve $\overline{e_i}$. Let $x_i$ be a variable with no literals in $c_j$ and let $c_j$ be the $k$th clause with no literals from $x_i$. Then for $l \in [n-2]$ the voter $v_j^l$ approves edge $(e_i)^l_k$. Let $\tilde{v}$ be one arbitrary voter from $\bigcup_{c_j \in C} V_j$. All voters apart from $\tilde{v}$ approve edge $e_C^1$ and $\tilde v$ approves $e_C^3$.
    \item We have $2n$ control voters $V^* = \{v^*_1, \dots v^*_{n}\}$ and $V^{**} = \{v^{**}_1, \dots v^{**}_{n}\}$. For each variable $x_i \in W$ and $l \in [n]$ the voter $v^*_l \in V^*$ approves edge ${(e_i^*)}_l$. The voter $v^{**}_i$ approves ${(e_i^*)}_i$. For $l \in [i-1]$ the voter $v^{**}_l$ approves edge ${(e_i^*)}_{n+l}$ and for $l \in [i+1, 2n]$ the voter $v^{**}_l$ approves edge ${(e_i^*)}_{n+l-1}$. All voters in $V^*$ and $V^{**}$ approve $e_C^2$.
  \end{itemize}
  As input matching we choose:
  \(M := \{{(e_i^*)}_{i} \mid i \in [n]\} \cup \{e_C^1\}.\)
  Note that all preferred matchings and the input matching are maximal and $M$ is approved by \textit{all} voters apart from $\tilde{v}$.
  The \textsc{wCW-Verification} instance is defined by graph $G$, voters $\bigcup_{c_j \in C} V_j \cup V^* \cup V^{**}$ and input matching $M$.

  Correctness:
  $[\Rightarrow]$
  Let $(X, C)$ be a Yes-Instance for \textsc{$3$-Sat}. Then there is an assignment for $X$ such that all clauses in $C$ are satisfied. Let $X^t$ be the variables that are assigned \textit{true} and $X^f$ the variables that are assigned \textit{false}. The matching
  \(M' := \{e_i \mid x_i \in X^t\} \cup \{\overline{e_i} \mid x_i \in X^f\} \cup \{e_C^2\}\)
  is a matching that is approved by all voters. Therefore it wins against $M$ in a plurality vote and $M$ is not a weak Condorcet winner.

  $[\Leftarrow]$
  Assume there is a matching $M' \neq M$ that wins against $M$ in a pairwise comparison against $M$. As $M$ is already approved by all but one voter, $M'$ must be approved by all voters. For a voter to have a positive utility in matching $M$, there must be at least one edge in the matching that is approved by the voter. Let $c_j$ be an arbitrary clause.
  We can note the following about $M'$:
  \begin{enumerate}
    \item The edge $e_C^2$ is in $M'$.
    \item For a clause $c_j$ at least one edge corresponding to a literal in $c_j$ must be in $M'$
  \end{enumerate}
  We now prove these statements.
  \begin{enumerate}
    \item We know that $M' \neq M$ and that all voters approve $M'$. Consider an edge in $M \setminus M'$. If this edge is $e_C^1$, then there must be at least one gadget corresponding to the variables in which another edge than ${(e_i^*)}_{i}$ is chosen. Let $x_i$ be a variable for which ${(e_i^*)}_{i} \notin M$. As ${(e_i^*)}_{i}$ is approved by $v^*_i$ and $v^{**}_i$, there must be another edge in $M'$ that is approved by these voters. Using the edges from the variable gadgets, the only way to create a matching that is approved by all voters in $V^*$ and $V^{**}$ is to choose $\{{(e_i^*)}_{i} \mid i \in [n]\}$. As this is not the case, the voters in $V^*$ and $V^{**}$ must be satisfied using the control gadget and thus $e_C^2 \in M'$.
    \item As $e_C^2 \in M'$, we know that $e_C^1 \notin M'$ and $e_C^3 \notin M'$. As $M'$ is approved by all voters, the voters in $\bigcup_{c_j \in C}V_j$ are satisfied by edges in the variable gadgets. It is not possible to satisfy all $n-2$ voters corresponding to clause $c_j$ using only the edges that do \textit{not} correspond to the literals in $c_j$. There are $n-3$ variables without literals in $c_j$ and each of the $n-2$ voters corresponding to $c_j$ approve a different edge in each of the gadgets. As only one edge per gadget can be in any matching, at most $n-3$ voters can be satisfied. As all $n-2$ voters corresponding to clause $c_j$ have positive utility, there must be at least one edge corresponding to a literal in $c_j$ in matching $M'$.
  \end{enumerate}
  Assigning \textit{true} to all variables $x_i$ with $e_i \in M'$ and \textit{false} to all variables with $\overline{e_i} \in M'$ is a satisfying assignment for the \textsc{$3$-Sat} instance.

  \medskip
  For \textsc{sCW-Verification} we use a similar reduction as for \textsc{wCW-Verification}. The only difference is, that we do not select a voter $\tilde{v}$ that does approve edge $e_C^3$ instead of $e_C^1$. Instead all voters in $\bigcup_{c_j \in C}V_j$ approve $e_C^1$. Thus, the input matching $M$ is approved by all voters and only another matching that is approved by all voters wins or ties against $M$. The construction of the fulfilling variable assignment for the \textsc{$3$-Sat} instance is analogous.

  \paragraph{$\kappa$-missing approval}
  For \textsc{wCW-Verification} under $\kappa$-missing approval we reduce from \textsc{Utilitarian Welfare}. This was shown to be NP-complete in \Cref{thm:approval_utilitarian}. The idea is to construct a matching that is approved by $k-1$ voters so that a matching with utilitarian welfare $k$ corresponds to a matching that wins in a majority vote. We use additional control voters to ensure that the added voters can only approve matchings that are not approved by any voter from the original graph.

  Construction:
  Given a graph $G$, a set of voters $V$ and a $k \in \mathbb{N}_0$, we construct graph $G'$ by extending $G$ as follows (see \Cref{fig:wCW_approval_graph,fig:wCW_approval_gadget} for an example):
  \begin{itemize}
    \item For each node $i \in W$ we add an edge $e_i$ that is only connected to that node.
    \item We add $3 \kappa$ control gadgets. The first $\kappa$ of them consist of four edges, forming a star and the last $2\kappa$ consist of two edges. We call the edges $e_{j}^1$ to $e_{j}^4$ for $j \in [\kappa]$ and $e_{j}^1$ to $e_{j}^2$ for $j \in [\kappa + 1, 3\kappa]$.
  \end{itemize}
  \begin{figure}[tbp]
    \centering

    \begin{subfigure}[t]{0.30\linewidth}
      \centering
      \begin{tikzpicture}[thick, main/.style = {draw, circle, minimum size=8mm}, node distance=20mm]

  \node[main] (v1) at (0,0) {};
  \node[main] (v2) at (2,0) {};
  \node[main] (v3) at (4,0) {};
  \node[main] (v4) at (0,2) {};
  \node[main] (v5) at (2,2) {};

  \node[main] (e1) [below=1cm of v1] {};
  \node[main] (e2) [below=1cm of v2] {};
  \node[main] (e3) [below=1cm of v3] {};
  \node[main] (e4) [above=1cm of v4] {};
  \node[main] (e5) [above=1cm of v5] {};

  \draw (v1) -- (v2);
  \draw (v2) -- (v3);
  \draw (v1) -- (v4);
  \draw (v4) -- (v5);
  \draw (v5) -- (v2);

  \draw[red] (v1) -- node[midway, fill=white] {$e_1$} (e1);
  \draw[red] (v2) -- node[midway, fill=white] {$e_2$} (e2);
  \draw[red] (v3) -- node[midway, fill=white] {$e_3$} (e3);
  \draw[red] (v4) -- node[midway, fill=white] {$e_4$} (e4);
  \draw[red] (v5) -- node[midway, fill=white] {$e_5$} (e5);

\end{tikzpicture}
      \caption{Graph $G'$ containing $G$ and added edges.}
      \label{fig:wCW_approval_graph}
    \end{subfigure}\hfill
    \begin{subfigure}[t]{0.65\linewidth}
      \centering
      \begin{tikzpicture}[thick, node distance=20mm, main/.style = {draw, circle, minimum size=8mm}]

  \node[main] (c1) {};

  \node[main] (c1a) at (90:2cm) {};
  \node[main] (c1b) at (180:2cm) {};
  \node[main] (c1c) at (270:2cm) {};
  \node[main] (c1d) at (360:2cm) {};

  \draw[red] (c1) -- node[midway, fill=white] {$\hat e_j^1$} (c1a);
  \draw (c1) -- node[midway, fill=white] {$\hat e_j^2$} (c1b);
  \draw (c1) -- node[midway, fill=white] {$\hat e_j^3$} (c1c);
  \draw (c1) -- node[midway, fill=white] {$\hat e_j^4$} (c1d);

  \node[main, right=5cm of c1] (c2b) {};
  \node[main, left of=c2b] (c2a) {};
  \node[main, right of=c2b] (c2c) {};
  \draw (c2a) -- node[midway, fill=white] {$\hat e_{j}^1$} (c2b);
  \draw[red] (c2b) -- node[midway, fill=white] {$\hat e_{j}^2$} (c2c);

\end{tikzpicture}
      \caption{Star control gadgets for $j \in [\kappa]$ (left) and path control gadgets for $j \in [\kappa+1, 3\kappa]$ (right).}
      \label{fig:wCW_approval_gadget}
    \end{subfigure}

    \caption{Construction for \textsc{wCW-Existence} under approval utilities.}
    \label{fig:wCW_approval_construction}
  \end{figure}
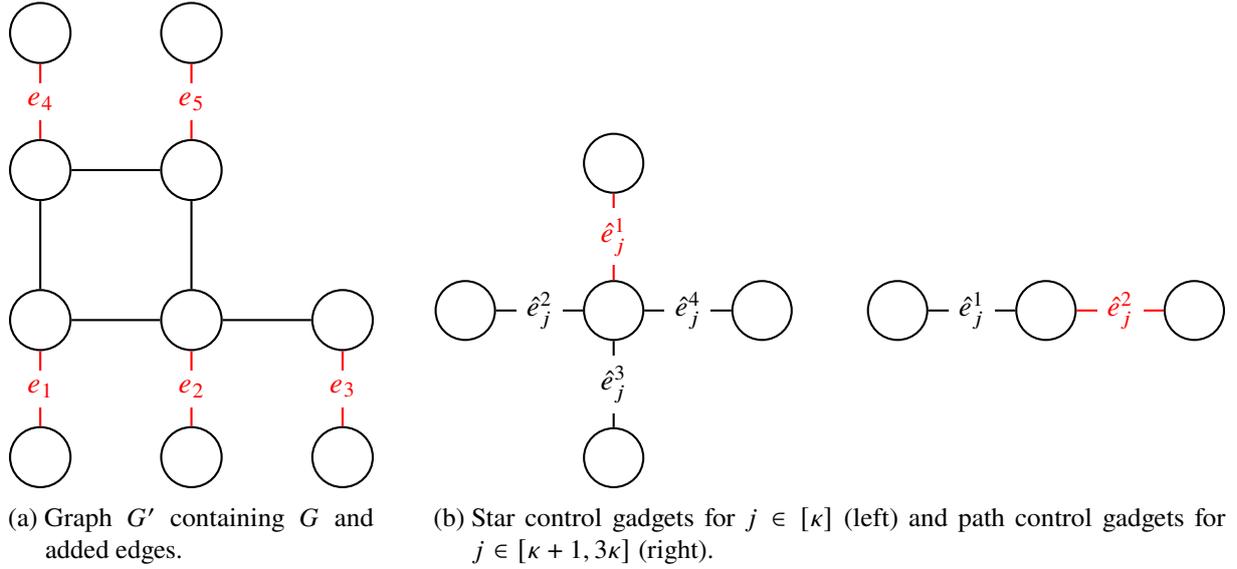

  The voters in $V$ have the same approvals as in the original instance.
  We add $4(|V|+1)+k - 1$ new voters with the following preferred matchings:
  \begin{itemize}
    \item We construct $k-1$ voters $V^*$. Those voters approve $\{e_i \mid i \in W\} \cup \{e_j^4 \mid j \in [\kappa]\}$
    \item We construct $|V| +1$ voters $V_1^*$ and $|V| +1$ voters $V_2^*$. The voters in $V_1^*$ approve $\{e_j^1 \mid j \in [2\kappa]\}$ and the voters in $V_2^*$ approve $\{e_j^2 \mid j \in [2\kappa]\}$.
    \item We construct $|V| +1$ voters $V_3^*$ and $|V| +1$ voters $V_4^*$. The voters in $V_3^*$ approve $\{e_j^1 \mid j \in [\kappa]\} \cup \{e_j^2 \mid j \in [\kappa + 1, 2\kappa]\} \cup \{e_j^1 \mid j \in [2 \kappa + 1, 3\kappa]\}$. The voters in $V_4^*$ approve $\{e_j^3 \mid j \in [\kappa]\} \cup \{e_j^2 \mid j \in [2\kappa + 1, 3\kappa]\}$.
  \end{itemize}
  We define our input matching as:
  \(M:= \{e_i \mid i \in W\} \cup \{e_j^1 \mid j \in [\kappa]\} \cup \{e_j^2 \mid j \in [\kappa + 1, 3\kappa]\}. \)
  Note that this input matching is approved by the voters in $V^* \cup V_1^* \cup V_2^* \cup V_3^* \cup V_4^*$.
  The \textsc{wCW-Verification} instance is defined by graph $G'$, voters $V \cup V^* \cup V_1^* \cup V_2^* \cup V_3^* \cup V_4^*$ and matching $M$.

  Correctness:
  $[\Rightarrow]$ Assume $(G, V, k)$ is a Yes-instance for \textsc{Utilitarian Welfare}. Then there is a matching $M'$ in $G$ that is approved by at least $k$ voters. Consider matching $M' \cup \{e_j^1 \mid j \in [\kappa]\} \cup \{e_j^2 \mid j \in [\kappa + 1, 3\kappa]\}$ in $G'$. This matching is approved by all voters in  $V_1^* \cup V_2^* \cup V_3^* \cup V_4^*$ and by at least $k$ voters from $V$. Therefore, it wins in a plurality vote against $M$ and $M$ is not a weak Condorcet Winner.

  $[\Leftarrow]$ Assume there is a matching $M'$ in $G'$ that wins against $M$. That is, it is approved by at least $4|V|+4+k$ voters. We can note the following about matching $M'$:
  \begin{enumerate}
    \item This matching is approved by $V_1^* \cup V_2^* \cup V_3^* \cup V_4^*$.
    \item For $j \in [2\kappa]$ either $e_j^1$ or $e_j^2$ is in $M'$.
    \item For $j \in [\kappa]$ $e_j^1 \in M'$ and for $j \in [\kappa + 1, 2\kappa]: j_j^2 \in M'$.
    \item Matching $M'$ is not approved by the voters in $V^*$.
    \item At least $k$ voters from $V$ approve $M'$.
  \end{enumerate}
  We now prove these statements.
  \begin{enumerate}
    \item Each of these sets contains $|V| + 1$ voters with the same preferred matchings. As there are only $|V|$ voters that do not already approve $M$, any matching that wins or ties in a plurality voting against $M$ must be approved by each voter in $V_1^* \cup V_2^* \cup V_3^* \cup V_4^*$.
    \item The voters in $V_1^*$ and $V_2^*$ approve only the edges in the control gadgets up to $2 \kappa$. They never share an edge. As for both sets of voters only $\kappa$ edges in their preferred matchings can be missing and they all must approve $M'$, in each of the control gadgets up to $2 \kappa$ the edge $e_j^1$ and $e_j^2$ must be in $M'$.
    \item As by (2) we know that edge $e_j^3$ is never in $M'$, we know that $\kappa$ edges from the preferred matching of voters in $V_4^*$ are missing. For them to still approve $M'$, it follows that for $j \in [2\kappa +1, 3\kappa]: e_j^2 \in M'$. Thus, $\kappa$ edges from the voters in $V_3^*$ are missing and in the gadgets from $j = 1$ to $j = 2\kappa$ the edge that is approved by voters in $V^*_3$ is in $M'$. That means that for $j \in [\kappa]$ $e_j^1 \in M'$ and for $j \in [\kappa + 1, 2\kappa]: j_j^2 \in M'$.
    \item We know that $M \neq M'$. From (1), (2) and (3) follows that $M'$ and $M$ are the same in the control gadgets. Thus, $M'$ must differ from $M$ in the original graph and the edges that were added to the original graph. As $M$ is a maximal matching, there must be an edge in $\{e_i \mid i \in W\}$ that is not in $M'$. As for the voters in $V^*$ already $\kappa$ edges are missing in the control gadgets, that means that there are at least $\kappa + 1$ edges missing. Thus, they do not approve matching $M'$.
    \item Matching $M'$ wins or ties in a plurality voting against $M$. We know that the $k$ voters in $V^*$ strictly prefer $M$ over $M'$. Therefore there must be at least $k$ voters in $V$ that strictly prefer $M'$ over $M$.
  \end{enumerate}
  From (5) follows that matching $M$, reduced on subgraph $G$ has a utilitarian welfare of $k$.

  \medskip
  For \textsc{sCW-Verification} the reduction is completely analogous to the reduction for \textsc{wCW-Verification} . The only difference is, that we construct $k$ voters $V^*$ instead of $k-1$. Thus, the input matching $M$ is approved by the $4|V+4$ control voters and $k$ additional voters. Using the same argumentation, a matching that wins or ties against $M$ then corresponds to a matching with utilitarian welfare $k$ in the original graph.
\end{proof}

\affinecondorcetexist*
\begin{proof}
  For both problems we reduce from \textsc{wPO-Verification}. This problem was proven to be coNP-complete in \Cref{thm:check_wsPO}.

  Construction for \textsc{wCW-Existence}:
  Consider an instance for \textsc{wPO-Verification}, defined by graph $G$, matching $M$ and voters $V$.
  To construct an instance for \textsc{wCW-Existence} we add two additional components, each consisting of a path of length two to graph $G$, to get graph $G'$. Let $e^a_1$ and $e^a_2$ be the edges of the first path and $e^b_1$ and $e^b_2$ be the edges of the second path.
  We construct $|V| + 3$ additional voters $V' = \{v'_1, \dots, v'_{|V|+1}, v_a, v_b\}$ with the following preferred matchings:
  \begin{itemize}
    \item For $i \in [|V|+1]$ voter $v'_i$ has preferred matching $M \cup \left \{e^a_1, e^b_1 \right \}$.
    \item Voter $v_a$ only approves edge $e^a_2$ and voter $v_b$ only approves edge $v_b^2$.
  \end{itemize}
  The \textsc{wCW-Existence} instance is defined by graph $G'$ and voters $V \cup V'$.
  A complete example for the construction can be seen in \Cref{fig:wCW_existence1,fig:wCW_existence2,fig:wCW_existence3}.

  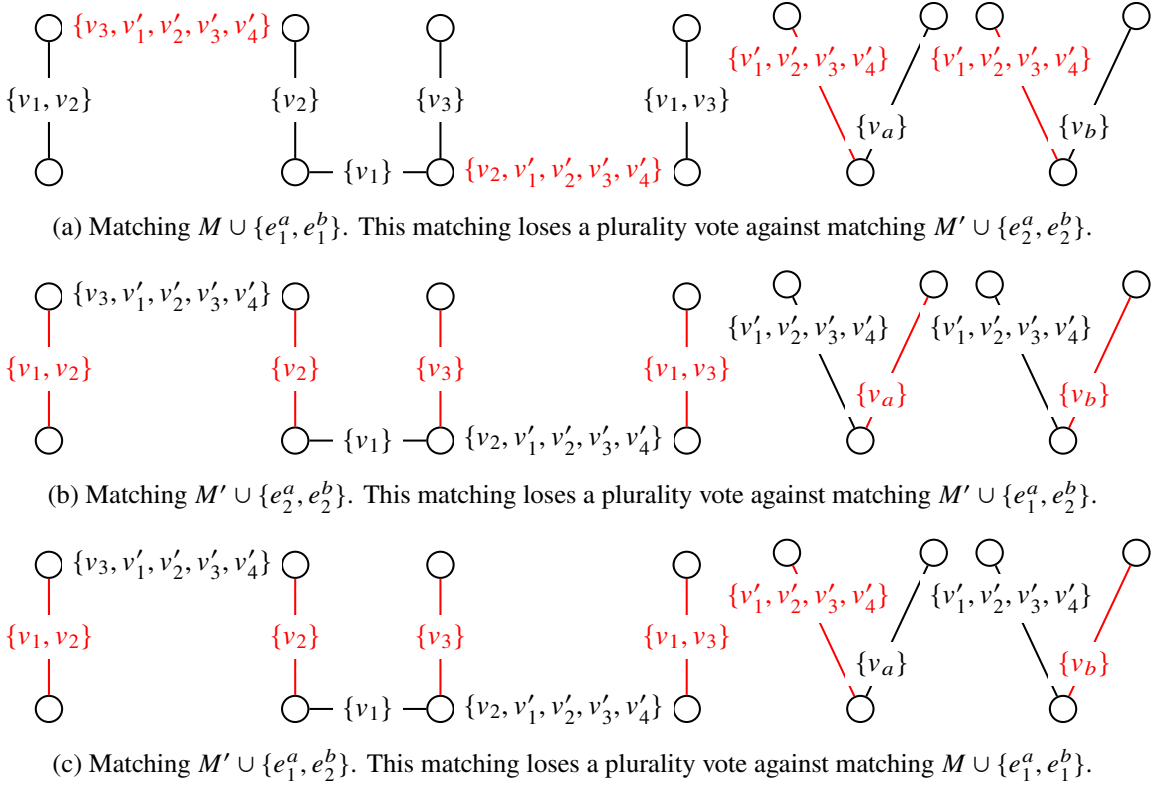
\begin{figure}[tbp]
    \centering

    \begin{subfigure}{0.95\linewidth}
      \centering
      \begin{tikzpicture}[node distance={20mm}, thick, main/.style = {draw, circle}]
    \node[main] (0) {};
    \node[main] (1) [above of=0] {};
    \node[main] (2) [right=3cm of 1] {};
    \node[main] (3) [below of=2] {};
    \node[main] (4) [right of=3] {};
    \node[main] (5) [above of=4] {};
    \node[main] (6) [right=3cm of 4] {};
    \node[main] (7) [above of=6] {};

    \node[main] (8) [right=2cm of 6] {};
    \node[main] (9) [angle={115}{of 8}] {};
    \node[main] (10) [angle={65}{of 8}] {};
    \node[main] (11) [right=2.4cm of 8] {};
    \node[main] (12) [angle={115}{of 11}] {};
    \node[main] (13) [angle={65}{of 11}] {};

    \draw[-] (0) edge node[midway, fill=white] {$\{v_1, v_2\}$} (1);
    \draw[-, red] (1) edge node[midway, fill=white] {$\{v_3, v'_1, v'_2, v'_3, v'_4\}$} (2);
    \draw[-] (2) edge node[midway, fill=white] {$\{v_2\}$}(3);
    \draw[-] (3) edge node[midway, fill=white] {$\{v_1\}$}(4);
    \draw[-] (4) edge node[midway, fill=white] {$\{v_3\}$}(5);
    \draw[-, red] (4) edge node[midway, fill=white] {$\{v_2, v'_1, v'_2, v'_3, v'_4\}$}(6);
    \draw[-] (6) edge node[midway, fill=white] {$\{v_1, v_3\}$}(7);

    \draw[-, red] (8) edge node[near end, fill=white] {$\{v'_1, v'_2, v'_3, v'_4\}$}(9);
    \draw[-] (8) edge node[near start, fill=white] {$\{v_a\}$}(10);
    \draw[-, red] (11) edge node[near end, fill=white] {$\{v'_1, v'_2, v'_3, v'_4\}$}(12);
    \draw[-] (11) edge node[near start, fill=white] {$\{v_b\}$}(13);
\end{tikzpicture}
      \caption{Matching $M \cup \{e^a_1, e^b_1\}$. This matching loses a plurality vote against matching $M' \cup \{e^a_2, e^b_2\}$.}
      \label{fig:wCW_existence1}
    \end{subfigure}

    \vspace{0.8em}

    \begin{subfigure}{0.95\linewidth}
      \centering
      \begin{tikzpicture}[node distance={20mm}, thick, main/.style = {draw, circle}]
    \node[main] (0) {};
    \node[main] (1) [above of=0] {};
    \node[main] (2) [right=3cm of 1] {};
    \node[main] (3) [below of=2] {};
    \node[main] (4) [right of=3] {};
    \node[main] (5) [above of=4] {};
    \node[main] (6) [right=3cm of 4] {};
    \node[main] (7) [above of=6] {};

    \node[main] (8) [right=2cm of 6] {};
    \node[main] (9) [angle={115}{of 8}] {};
    \node[main] (10) [angle={65}{of 8}] {};
    \node[main] (11) [right=2.4cm of 8] {};
    \node[main] (12) [angle={115}{of 11}] {};
    \node[main] (13) [angle={65}{of 11}] {};

    \draw[-, red] (0) edge node[midway, fill=white] {$\{v_1, v_2\}$} (1);
    \draw[-] (1) edge node[midway, fill=white] {$\{v_3, v'_1, v'_2, v'_3, v'_4\}$} (2);
    \draw[-, red] (2) edge node[midway, fill=white] {$\{v_2\}$}(3);
    \draw[-] (3) edge node[midway, fill=white] {$\{v_1\}$}(4);
    \draw[-, red] (4) edge node[midway, fill=white] {$\{v_3\}$}(5);
    \draw[-] (4) edge node[midway, fill=white] {$\{v_2, v'_1, v'_2, v'_3, v'_4\}$}(6);
    \draw[-, red] (6) edge node[midway, fill=white] {$\{v_1, v_3\}$}(7);

    \draw[-] (8) edge node[near end, fill=white] {$\{v'_1, v'_2, v'_3, v'_4\}$}(9);
    \draw[-, red] (8) edge node[near start, fill=white] {$\{v_a\}$}(10);
    \draw[-] (11) edge node[near end, fill=white] {$\{v'_1, v'_2, v'_3, v'_4\}$}(12);
    \draw[-, red] (11) edge node[near start, fill=white] {$\{v_b\}$}(13);
\end{tikzpicture}
      \caption{Matching $M' \cup \{e^a_2, e^b_2\}$. This matching loses a plurality vote against matching $M' \cup \{e_1^a, e_2^b\}$.}
      \label{fig:wCW_existence2}
    \end{subfigure}

    \vspace{0.8em}

    \begin{subfigure}{0.95\linewidth}
      \centering
      \begin{tikzpicture}[node distance={20mm}, thick, main/.style = {draw, circle}]
    \node[main] (0) {};
    \node[main] (1) [above of=0] {};
    \node[main] (2) [right=3cm of 1] {};
    \node[main] (3) [below of=2] {};
    \node[main] (4) [right of=3] {};
    \node[main] (5) [above of=4] {};
    \node[main] (6) [right=3cm of 4] {};
    \node[main] (7) [above of=6] {};

    \node[main] (8) [right=2cm of 6] {};
    \node[main] (9) [angle={115}{of 8}] {};
    \node[main] (10) [angle={65}{of 8}] {};
    \node[main] (11) [right=2.4cm of 8] {};
    \node[main] (12) [angle={115}{of 11}] {};
    \node[main] (13) [angle={65}{of 11}] {};

    \draw[-, red] (0) edge node[midway, fill=white] {$\{v_1, v_2\}$} (1);
    \draw[-] (1) edge node[midway, fill=white] {$\{v_3, v'_1, v'_2, v'_3, v'_4\}$} (2);
    \draw[-, red] (2) edge node[midway, fill=white] {$\{v_2\}$}(3);
    \draw[-] (3) edge node[midway, fill=white] {$\{v_1\}$}(4);
    \draw[-, red] (4) edge node[midway, fill=white] {$\{v_3\}$}(5);
    \draw[-] (4) edge node[midway, fill=white] {$\{v_2, v'_1, v'_2, v'_3, v'_4\}$}(6);
    \draw[-, red] (6) edge node[midway, fill=white] {$\{v_1, v_3\}$}(7);

    \draw[-, red] (8) edge node[near end, fill=white] {$\{v'_1, v'_2, v'_3, v'_4\}$}(9);
    \draw[-] (8) edge node[near start, fill=white] {$\{v_a\}$}(10);
    \draw[-] (11) edge node[near end, fill=white] {$\{v'_1, v'_2, v'_3, v'_4\}$}(12);
    \draw[-, red] (11) edge node[near start, fill=white] {$\{v_b\}$}(13);
\end{tikzpicture}
      \caption{Matching $M' \cup \{e_1^a, e_2^b\}$. This matching loses a plurality vote against matching $M \cup \{e^a_1, e^b_1\}$.}
      \label{fig:wCW_existence3}
    \end{subfigure}

    \caption{Example for \textsc{wCW-Existence}.}
    \label{fig:wCW_existence}
  \end{figure}

  Correctness for \textsc{wCW-Existence}:
  $[\Rightarrow]$
  Consider a No-instance for \textsc{wPO-Verification} defined by $(G, V, M)$. Thus there is a matching $M'$ s.th.\ all voters in $V$ have a strictly increased utility compared to $M$. In the \textsc{wCW-Existence} instance the matching $M' \cup \left \{e^a_2, e^b_2 \right \}$ wins in a pairwise comparison against $M \cup \left \{e^a_1, e^b_1 \right \}$, as the $|V|$ voters in $V$ and the two voters $v_a$ and $v_b$ strictly prefer $M' \cup \left \{e^a_2, e^b_2 \right \}$ over $M \cup \left \{e^a_1, e^b_1 \right \}$ while only the $|V|+1$ voters $v'_1$ to $v'_{|V|+1}$ strictly prefer $M \cup \left \{e^a_1, e^b_1 \right \}$ over $M' \cup \left \{e^a_2, e^b_2 \right \}$. Therefore $M \cup \{e^a_1, e^b_1\}$ is not a weak Condorcet winner.

  Consider any matching $\tilde M \neq M \cup \left \{e^a_1, e^b_1 \right \}$. As $M \cup \left \{e^a_1, e^b_1 \right \}$ is a maximal matching and $\tilde M \neq M \cup \left \{e^a_1, e^b_1 \right \}$, there is at least one edge $e \in M \cup \left \{e^a_1, e^b_1 \right \} \setminus \tilde M$. Let $E_e$ be the set of edges that are adjacent to $e$.
  Let $\tilde M' := \tilde M \cup \{e\} \setminus E_e$. $\tilde M'$ fulfills the matching property, as only edge $e$ is added to $\tilde M$ and all edges that are adjacent to $e$ are removed. All voters $v'_1$ to $v'_{|V|+1}$ strictly prefer $\tilde M'$ over $\tilde M$ as they have $M \cup \left \{e^a_1, e^b_1 \right \}$ as their preferred matching and gain utility from edge $e$.
  Assume $e$ is in the original graph $G$. At most $|V|$ voters strictly prefers $\tilde M$ over $\tilde M'$, as voters in $V$ approve edges from $E_e$. Now assume $e$ is not in $G$ which means that $e \in \left \{e^a_1, e^b_1 \right \}$. In this case at most one voter strictly prefers $\tilde M$ over $\tilde M'$, as $E_e$ only consists of one edge that is either approved by $v_a$ or by $v_b$.
  It follows that $\tilde M'$ wins a plurality voting against $\tilde M$. Therefore $\tilde M$ is not a weak Condorcet winner.
  As $M \cup \left \{e^a_1, e^b_1 \right \}$ is not a Condorcet winner and any matching $\tilde M \neq M$ is not a Condorcet winner, there exists no Condorcet winner.

  $[\Leftarrow]$
  Consider now a Yes-instance for \textsc{wPO-Verification}, defined by $(G, V, M)$. Thus, matching $M$ is strictly Pareto optimal. That is, for every $M' \neq M$, at most $|V|-1$ voters from $V$ strictly prefer $M'$ over $M$. It follows that in the \textsc{wCW-Existence} instance at most $|V|+1$ voters can prefer any matching over $M \cup \left \{e^a_1, e^b_1 \right \}$, as voters $v'_1$ to $v'_{|V|+1}$ always have maximum utility in $M \cup \left \{e^a_1, e^b_1 \right \}$.At most $|V|-1$ voters from $V$ and the two voters $v_a$ and $v_b$ can have increased utility in any other matching. As at least $|V|+1$ voters ($v'_1, \dots, v'_{|V|+1}$) strictly prefer $M$, $M$ wins or ties in a pairwise comparison against every other matching and is a weak Condorcet winner.

  Construction for \textsc{sCW-Existence}:
  Consider an instance for \textsc{wPO-Verification}, defined by graph $G$, matching $M$ and voters $V$.
  To construct an instance for \textsc{sCW-Existence} we add $|V|$ additional voters $V' = \{v'_1, \dots, v'_{|V|}\}$. For each voter $v' \in V': M^*(v') = M$ holds.
  The \textsc{sCW-Existence} instance is defined by graph $G$ and voters $V \cup V'$.

  Correctness for \textsc{sCW-Existence}:
  $[\Rightarrow]$
  Consider a No-instance for \textsc{wPO-Verification}. There is a matching $M'$ such that all voters in $V$ have a strictly increased utility compared to input matching $M$. The matching $M'$ ties in a pairwise comparison against $M$, as the $|V|$ voters in $V$ strictly prefer $M'$ over $M$ and the $|V|$ voters in $V'$ strictly prefer $M$ over $M'$. Therefore $M$ is not a strong Condorcet winner.

  Consider any matching $\tilde M \neq M$. As $M$ is a maximal matching and $\tilde M \neq M$, there is at least one edge $e \in M \setminus \tilde M$. Let $E_e$ be the set of edges that are adjacent to $e$.
  Let $\tilde M' := \tilde M \cup \{e\} \setminus E_e$. The set $\tilde M'$ fulfills the matching property, as only edge $e$ is added to matching $M$ and all edges that are adjacent to $e$ are removed. As all voters in $V'$ have $M$ as their preferred matching, at least $|V|$ voters prefer $\tilde M'$ over $\tilde M$ as they get utility from edge $e$ and don't loose utility from edges $E_e$. At most $|V|$ voters strictly prefer $\tilde M$ over $\tilde M'$, as only original voters approve edges from $E_e$. It follows that $\tilde M'$ at least ties against $\tilde M$. Therefore $\tilde M$ is not a strong Condorcet winner.

  As $M$ is not a Condorcet winner and any matching $\tilde M \neq M$ is not a Condorcet winner, there is no Condorcet winner.

  $[\Leftarrow]$
  Consider now a Yes-instance for \textsc{wPO-Verification}. Thus, matching $M$ is weakly Pareto optimal. That is for every $M' \neq M$, at most $|V|-1$ voters from $V$ strictly prefer $M'$ over $M$. As all voters in $V'$ by definition strictly prefer $M$ over every other matching, at least $|V|$ voters prefer $M$ over $M'$. Thus, $M$ wins a plurality vote against every other matching $M'$ and is a strong Condorcet winner.
\end{proof}

\approvalsCWexist*
\begin{proof}~
  \paragraph{One-edge approval}
  For one-edge approval we reduce from \textsc{$3$-Sat}. The idea is to use the same construction for the graph and the voters as in \Cref{thm:check_wsCW}. We show that the matching $M$ that we defined for that construction is a strong Condorcet winner if and only if there is no satisfying assignment for the \textsc{$3$-Sat} instance and that otherwise no strong Condorcet winner exists.

  Construction:
  We use the same construction as in \Cref{thm:check_wsCW} to get a graph $G'$ and voters $\bigcup_{c_j \in C} V_j \cup V^* \cup V^{**}$.

  Correctness:
  $[\Rightarrow]$ Let $(X, C)$ be a Yes-instance for \textsc{$3$-Sat}. That is, there is an assignment for $X$ such that all clauses in $C$ are satisfied. Let $X^t$ be the variables that are assigned \textit{true} and $X^f$ the variables that are assigned \textit{false}. The matchings
  \(M := \{{(e_i^*)}_{i} \mid i \in [n]\} \cup \{e_C^1\}\)
  and
  \(M' := \{e_i \mid x_i \in X^t\} \cup \{\overline{e_i} \mid x_i \in X^f\} \cup \{e_C^2\}\)
  are both approved by all voters. Thus, they tie in a pairwise comparison and there is no matching that wins in a pairwise comparison against them. Therefore there is no strong Condorcet winner.

  $[\Leftarrow]$ Let $(X, C)$ be a No-instance for \textsc{$3$-Sat}. That is, there is \textit{no} assignment for $X$ such that all clauses in $C$ are satisfied.
  In graph $G'$ the matching $M := \{{(e_i^*)}_{i} \mid i \in [n]\} \cup \{e_C^1\}$ is approved by all voters. For that matching to \textit{not} be a strong Condorcet winner there must be another matching that is also approved by all voters. Following the same arguments as in the proof of \Cref{thm:check_wsCW} such a matching would encode a satisfying assignment for the \textsc{$3$-Sat} instance.

  \paragraph{$\kappa$-missing approval}
  For $\kappa$-missing approval we reduce from \textsc{Utilitarian Welfare} which was shown to be NP-complete in \Cref{thm:approval_utilitarian}. The idea is to use a similar construction as in \Cref{thm:check_wsCW} and add a gadget that makes the utilitarian welfare in the original instance non-unique. We show that the matching $M$ that we defined for the construction is a strong Condorcet winner if and only if there is no matching with a utilitarian welfare of at least $k$ and that otherwise no strong Condorcet winner exists, as the matching with highest utilitarian welfare in the original graph is not unique, even if it was unique in the original graph.

  Construction:
  Given a graph $G = (W, E)$, a set of voters $V$ and a $k \in \mathbb{N}_0$, we construct a graph $G'$ by extending $G$ as follows.
  \begin{itemize}
    \item For each node $i \in W$ we add an edge $e_i$ that is only connected to that node.
    \item We add $3 \kappa$ control gadgets. The first $\kappa$ of them consist of four edges, forming a star and the last $2\kappa$ consist of two edges. We call the edges $e_{j}^1$ to $e_{j}^4$ for $j \in [\kappa]$ and $e_{j}^1$ to $e_{j}^2$ for $j \in [\kappa+1, 3\kappa]$.
    \item We add a gadget consisting of two edges $e_C^1$ and $e_C^2$.
  \end{itemize}
  The voters in $V$ have the same approvals as in the original instance.
  We construct $4(|V|+1) + k + 2$ additional voters.
  \begin{itemize}
    \item We construct $k$ voters $V^*$. Those voters approve $\{e_i \mid i \in W\} \cup \{e_j^4 \mid j \in [\kappa]\} \cup \{e_C^1\}$
    \item We construct $|V| + 1$ voters $V_1^*$ and $|V| +1$ voters $V_2^*$. The voters in $V_1^*$ approve $\{e_j^1 \mid j \in [2\kappa]\}$ and the voters in $V_2^*$ approve $\{e_j^2 \mid j \in [2\kappa]\}$.
    \item We construct $|V| + 1$ voters $V_3^*$ and $|V| +1$ voters $V_4^*$. The voters in $V_3^*$ approve $\{e_j^1 \mid j \in [\kappa]\} \cup \{e_j^2 \mid j \in [\kappa + 1, 2\kappa]\} \cup \{e_j^1 \mid j \in [2\kappa +1, 3\kappa]\}$. The voters in $V_4^*$ approve $\{e_j^3 \mid j \in [\kappa]\} \cup \{e_j^2 \mid j \in [2\kappa +1, 3\kappa]\}$.
    \item We construct two voters $v^*_1$ and $v^*_2$. Voter $v^*_1$ approves edge $e_C^1$ and voter $v^*_2$ approves edge $e_C^2$.
  \end{itemize}

  Correctness:
  $[\Rightarrow]$ Assume $(G, V, k)$ is a Yes-instance for \textsc{Utilitarian Welfare}. Then the matching with highest utilitarian welfare in the original graph is approved by at least $k$ voters. Let $M$ be a matching with maximum utilitarian welfare and let that welfare be $k'$. Consider the following two matchings:
  \(M_A := M \cup \{e_j^1 \mid j \in [\kappa]\} \cup \{e_j^2 \mid j \in [\kappa + 1, 3\kappa]\} \cup \{e_C^1\}\) and
  \(M_B := M \cup \{e_j^1 \mid j \in [\kappa]\} \cup \{e_j^2 \mid j \in [\kappa + 1, 3\kappa]\} \cup \{e_C^2\}.\)

  Both matchings are approved by all voters in  $V_1^* \cup V_2^* \cup V_3^* \cup V_4^*$ and by at least $k$ voters from $V$. If $\kappa = 0$ they are also approved either $v^*_1$ or $v^*_2$ and if $\kappa > 0$ they are both approved by both voters. Thus, they have both the same utilitarian welfare of $4|V+1| + k' + 1$ (if $\kappa = 0$) or $4|V+1| + k' + 2$ (if $\kappa > 0$).

  We now want to show that there is no matching that wins against $M_A$ and $M_B$.
  Assume there is a matching $\tilde M$ with strictly higher utilitarian welfare. We can note the following about $\tilde M$:
  \begin{enumerate}
    \item For $i \in [4]$ the matching $\tilde M$ is approved by $V_i^*$
    \item If it is approved by the voters in $V^*$, it is not approved by any voter in $V$.
    \item If $\kappa = 0$ and it is approved by  $v^*_1$, then it is not approved by  $v^*_2$. If $\kappa > 0$ it is approved by both.
  \end{enumerate}
  We now prove these statements:
  \begin{enumerate}
    \item There are at most $|V| + 1$ voters that do not approve $M_A$ and $M_B$. As each set $V_i^*$ contains $|V| + 1$ voters, it is not possible to have a matching with higher utility than $M_A$ and $M_B$ that is not approved by the voters in these sets.
    \item From (1) we know that the matching is approved by all voters in $V_i^*$. Using the same line of arguments as in\Cref{thm:check_wsCW}, it follows that  $\{e_j^1 \mid j \in [\kappa]\} \cup \{e_j^2 \mid j \in [\kappa + 1, 3\kappa]\} \in \tilde M$. Thus, there are already $\kappa$ edges from the preferred matchings of the voters in $V^*$ missing. For these voters to still approve $\tilde M$ all other edges must be in $\tilde M$, which blocks all edges in the original graph and therefore blocks all edges corresponding to the original voters.
    \item If $\kappa = 0$ then the complete preferred matching must be in $\tilde M$ for a voter to approve that matching. As each of $v^*1$ and $v^*_2$ approve only one single edge and they are adjacent to each other, it is not possible for them to approve the same matching. For $\kappa > 0$ they both approve every matching.
  \end{enumerate}
  From statement (2) follows that either the $k$ voters in $V^*$ or some voters in $V$ approve the matching. As matching $M$ is the matching with highest utilitarian welfare $k'$ in the original instance, a matching with higher utilitarian welfare than $4|V+1| + k' + 1$ (when $\kappa = 0$) or $4|V+1| + k' + 2$ (when $\kappa > 0$) cannot exist.

  Thus, the matchings $M_A$ and $M_B$ are both matchings with maximum utilitarian welfare and there is no strong Condorcet winner.

  $[\Leftarrow]$ Assume $(G, V, k)$ is a No-instance for \textsc{Utilitarian Welfare}. Then all matchings in the original graph have an utilitarian welfare of at most $k-1$. Consider the matching
  \(M:= \{e_i \mid i \in W\} \cup \{e_j^1 \mid j \in [\kappa]\} \cup \{e_j^2 \mid j \in [\kappa + 1, 3\kappa]\} \cup \{e_C^1\}\).
  This matching has a utilitarian welfare of $4|V+1| + k + 1$ (when $\kappa = 0$) or $4|V+1| + k' + 2$ (when $\kappa > 0$).
  Assume there is a matching $M'$ with the same or higher utilitarian welfare. Following the same arguments as in \Cref{thm:check_wsCW} we know that this matching must have the same edges as $M$ in the gadgets. As we have already shown, it is not possible for a matching to be approved by the voters in $V^*$ and any voter in $V$. As there is no matching with utilitarian welfare of at least $k$ in the original instance, the matching $M'$ must be approved by the voters in $V^*$. Thus, all edges that are approved by the voters in $V^*$ are in $M'$ and $M' = M$.
  As there is no matching that wins or ties against $M$, this matching is a strong Condorcet winner.
\end{proof}

\end{document}